\documentclass[11pt, oneside]{article}

\usepackage{vmargin}   %Ustawia marginesy.
\setmarginsrb{24mm}{0mm}{24mm}{10mm}{10mm}{10mm}{10mm}{10mm}

\usepackage{graphicx,epsfig}
\usepackage[cp1250]{inputenc}
\usepackage{amsmath}
\usepackage{amssymb}
\usepackage{amsthm}
\usepackage{cite}
\usepackage{colortbl}
\usepackage{hhline}
\usepackage{enumerate}

\usepackage{multirow}
\usepackage{pdflscape}
\usepackage{tikz}

\newtheorem{thm}{Theorem}[section]
\newtheorem{lem}[thm]{Lemma}
\newtheorem{cor}[thm]{Corollary}
\newtheorem{prop}[thm]{Proposition}

\theoremstyle{definition}
\newtheorem{rem}[thm]{Remark}
\newtheorem{defn}[thm]{Definition}
\newtheorem{ex}[thm]{Example}
\newtheorem{con}[thm]{Conjecture}

\newcommand\diag[3]{%
	\multirow{-#1}{#2}{%
		\hskip-\tabcolsep\hskip-0.1mm
		$\vcenter{\vskip-0.5mm
			\begin{tikzpicture}
				\node[minimum width={#2+2\tabcolsep+1.0mm}, minimum height={#1\baselineskip+0.6mm}] (box) {};
				\draw (box.south east) -- (box.north west);
				\node[anchor=south west] at (box.south west) {#3};
			\end{tikzpicture}
		}$
		\hskip-\tabcolsep
	}
}

\newcommand{\pow}[1]{\mathcal{P}(#1)}
\newcommand{\powk}[2]{\mathcal{P}_{#1}(#2)}

\newcommand{\Pol}{\mathcal{Z}} 
\newcommand{\B}{\mathcal{B}} 
\newcommand{\bfg}{\boldsymbol{g}} 
\newcommand{\G}{\Gamma} 
\newcommand{\Gm}{\min \G}
\newcommand{\GPPD}{\G=\G(\Pi,\Pol,\Delta)}

\newcommand{\NN}{\mathbb{N}}
\newcommand{\NO}{\mathbb{N}_0}
\newcommand{\RR}{\mathbb{R}}
\newcommand{\KK}{\mathbb{K}}

\newcommand{\W}{\bar{w}} 
\newcommand{\V}{\bar{v}} 
\newcommand{\U}{\bar{u}} 
\newcommand{\E}[1]{\bar{e}^{(#1)}} 
\newcommand{\supp}[1]{\mathrm{supp}(#1)}
\newcommand{\sgn}{\mathrm{sgn}}
\newcommand{\Span}{\mathrm{span}}

\newcommand{\seq}{\subseteq}
\newcommand{\set}[1]{\{ #1\}}

\newcommand{\ord}[2]{Ord(#1,#2)}
\newcommand{\pord}[2]{Ord^*(#1,#2)}

\begin{document}

\title{Access Structures Determined\\ by Uniform Polymatroids}

\author{Renata~Kawa\\
{\footnotesize Faculty of Science and Technology,} \\
{\footnotesize Jan D\l ugosz University, Cz\c estochowa, Poland}\\ 
{\footnotesize e-mail: (r.kawa@ujd.edu.pl)}\\
{\footnotesize ORCID: 0000-0002-3224-7476}
\\[5mm]
Mieczys\l aw~Kula\\
{\footnotesize Institute of Mathematics}\\
{\footnotesize University of Silesia, Katowice, Poland}\\ 
{\footnotesize ORCID: 0000-0001-5743-3809}
}

\date{ }

\maketitle

\begin{abstract}
A secret sharing scheme is a method of sharing a secret key among a finite set of participants in such a way that only certain specified subsets of participants can compute the key. The access structure of a secret sharing scheme is the family of these subsets of participants which are able to recover the secret. If the length in bits of every share is the same as the length of the secret, then the scheme is called ideal.
An access structure is said to be multipartite, if the set of participants is divided into several parts and all participants in the same part play an equivalent role.
The search for ideal secret sharing schemes for some special interesting families of multipartite access structures, has been carried out by many authors. 
In this paper a new concept of study of ideal access structures is proposed.  
We do not consider special classes of access structures defined by imposing certain prescribed assumptions, but we investigate all access structures obtained from uniform polymatroids using the method developed by Farr\`as, Mart\'i-Farr\'e and Padr\'o (cf.  Theorem \ref{thm:mingamma} and Definition  \ref{defn:constr:acc:str} below). 
They satisfy necessary condition to be ideal, i.e., they are matroid ports.
Moreover some objects in this family can be useful for the applications of secret sharing.
The choice of uniform polymatroids is motivated by the fact that each such polymatroid defines ideal access structures. The method presented in this article is universal and can be continued with other classes of polymatroids in further similar studies. 
Here we are especially interested in hierarchy of participants determined by the access structure and we distinguish two main classes: they are compartmented and hierarchical access structures. The vast majority of papers discussing hierarchical access structures consider access structures which are compartment or totally hierarchical.
The main results are summarized in Section \ref{sec:hier:pord}, which presents situations where partial hierarchy properties may arise.
In particular, hierarchical orders of obtained structures are described. It is surprising, that the hierarchical orders of access structures obtained from uniform polymatroids are flat, i.e., every chain has at most 2 elements. The ideality of some families of hierarchical access structures is proved in Section \ref{sec:ideal:acc:str}.
\end{abstract}

\noindent
\textbf{Keywords}
secret sharing - multipartite access structure - ideal access structure - partially hierarchical access structure - uniform polymatroid.

%================================================================
\section{Introduction}
%================================================================

A secret sharing scheme is a method of sharing a secret piece of data among a finite 
set of participants in such a way that only certain specified subsets of participants 
can compute the secret data. Secret sharing was originally introduced by Blakley 
\cite{gB} and Shamir \cite{aS} independently in 1979 as a solution for safeguarding 
cryptographic keys, but nowadays it is used in many cryptographic protocols.

Let $P$ be a finite set of participants and let $p_0\notin P$ be a special participant 
called the \emph{dealer}. Given a secret, the dealer computes the shares and distributes 
them secretly to the participants, so that no participant knows the share given to 
another one. It is required that only certain \emph{authorized} subsets of $P$ can 
recover the secret by pooling their shares together. It is easily seen that the family 
$\Gamma$ of all authorized sets, called an \emph{access structure}, is monotone 
increasing, which means that any superset of an authorized subset is also authorized. To 
avoid abnormal situations, we assume that $\emptyset \notin \Gamma$ and $P \in \Gamma$. 
If every unauthorized set of participants cannot reveal any information about the 
secret, regardless of the computational power available, then the secret sharing scheme 
is said to be \emph{perfect}. Such a scheme can be considered as unconditionally secure.

Ito, Saito, Nishizeki \cite{ISN} and Benaloh, Leichter \cite{BL} independently  proved, 
in a constructive way, that every monotone increasing family of subsets of $P$ admits a 
perfect secret sharing scheme. Therefore, every monotone increasing family of subsets of 
$P$ is referred to as an access structure. Obviously, every access structure is uniquely 
determined by the family of its minimal sets. An access structure is said to be 
\emph{connected} if every participant in $P$ is a member of a certain minimal authorized 
set. If an access structure is not connected, then every participant which does not 
belong to any minimal authorized set is called \emph{redundant} because its share is 
never necessary to recover the secret. 

Given a secret sharing scheme, let $S_0$ be the set of all possible secrets and let 
$S_p$ be the set of all possible values of shares that can be assigned to the 
participant $p$ for every $p\in P$. One can show that for every perfect secret sharing 
scheme the size of $S_0$ is not greater than the size of $S_p$ for all $p\in P$. A 
perfect secret sharing scheme is called \emph{ideal} if $|S_0|=|S_p|$ for all $p\in P$. 
In other words, the length in bits of every share is the same as the length of the 
secret. Shamir's threshold schemes \cite{aS} are the best known examples of ideal secret 
sharing schemes. The secret sharing schemes constructed for a given access structure 
in \cite{ISN} and \cite{BL} are very far from being ideal because the length of the 
shares grows exponentially with the number of participants. An access structure is said 
to be \emph{ideal}, if it is the access structure of an ideal secret sharing scheme. 
An access structure is said to be \emph{multipartite} if the set of participants 
is divided into several blocks which are pairwise disjoint and participants in 
individual blocks are equivalent (precise definition can be found in subsection 
\ref{ssec:multipartite}).
The study of multipartite access structures was initiated by Kothari \cite{sK}, 
who posed the open problem of constructing ideal hierarchical secret sharing schemes, 
and by Simmons \cite{gS}, who introduced the multilevel and compartmented access 
structures. This approach, developed by many authors, provides a very effective tool for 
describing structures in a compact way, by using a few conditions that are independent 
of the total number of participants. 

The characterization of ideal access structures is one of the main open problems in the 
secret sharing theory. This problem seems to be extremely difficult and only some 
particular results are known. In many papers the authors consider some specific classes 
of access structures with prescribed properties and try to check whether these 
structures are ideal. Most of the results obtained are based on the connections between 
ideal secret sharing schemes and matroids dis\-covered by Brickell \cite{eB} and 
Brickell and Davenport \cite{BD}. Later, the use of polymatroids proposed by Farr\`as, 
Mart\'i-Farr\'e, Padr\'o in \cite{FMFP} provided a new tool for studying ideal 
multipartite access structures. 

A concise review of the results contained in the literature can be found in the papers 
\cite{FMFP} - \cite{FPXY}. Since ideal access structures are known to be matroid ports, 
it seems quite natural to look for ideal access structures among matroid ports. 
Given a specific class of polymatroids, one can take all multipartite access structures 
determined by these polymatroids and investigate their properties. This approach ensures 
that the objects under consideration satisfy necessary condition to be ideal, i.e. they 
are matroid ports (cf. Theorem \ref{thm:necess:cond} below). The ideality can be 
established on the base of properties of parti\-cular polymatroids. In this paper the 
study is restricted to uniform polymatroids. This choice is motivated by the fact that 
each such polymatroid defines a family of ideal access structure (cf. Remark 
\ref{rem:uniform:ideal}). But the method presented here is universal and can be 
continued with other classes of polymatroids in further similar studies (cf. \cite{mK}).

The relations between ideal access structures and matroids discovered by 
Brickell and Davenport are recalled here in Theorem \ref{thm:necess:cond} and Theorem 
\ref{thm:suffic:cond}. A short introduction to matroids and polymatroids and their 
relation to access structures are presented in Subsection \ref{sec:polymatroids}. It 
follows from Theorem \ref{thm:mingamma} by Farr\`as, Mart\'i-Farr\'e and Padr\'o 
\cite{FMFP} that every polymatroid with the ground set $J$ and a monotone increasing 
family of subsets of $J$ which is compatible with the polymatroid determine a unique 
access structure which is a matroid port. The details are described in Definition 
\ref{defn:constr:acc:str}. In Subsection \ref{sec:uniform} some relations between  
uniform polymatroids $\Pol=(J,h,\bfg)$ and monotone increasing families $\Delta\seq 
\pow{J}\setminus\set{\emptyset}$ are presented. We prove several technical properties 
which are useful in the next sections.

In this paper, we focus on the classification of multipartite access structures $\GPPD$ 
determined by uniform polymatroids $\Pol$ and monotone increasing families $\Delta$  in 
a set of participants divided into a partition $\Pi$. We examine hierarchical order 
among the participants induced by the obtained access structure. In the third section 
we present several conditions that polymatroid $\Pol$ and monotone increasing family 
$\Delta$ must meet when the structure $\GPPD$ is (weakly) hierarchical. It turns out 
that the existence of hierarchically comparable  blocks imposes strong restrictions on 
the increment sequence $\bfg$ of the polymatroid. 

At the beginning of the fourth section, which together with the fifth section contains 
the main results of the paper, those conditions are used to prove Theorems 
\ref{thm:extreme:columns} and \ref{thm:last:row} which shows that most of access 
structures obtained from uniform polymatroids are compartmented. Then the exact 
hierarchy in some special access structures is examined in Theorems \ref{thm:g2:zero} 
and \ref{thm:last:column} - \ref{thm:singleton}. Moreover, we prove in Theorem 
\ref{thm:height} that the maximal length of chains in such hierarchical access 
structures is equal to 1. This fact seems quite surprising, because 
for other polymatroids one can construct hierarchical access structures 
with chains of arbitrary length. For instance, such constructions can be found in 
\cite{FP}, \cite{FPXY}, \cite{mK}, \cite{tT} and others. As was mentioned above, every 
uniform polymatroid determines some ideal access structures, but the question is whether 
all access structures determined by uniform polymatroids are ideal.  
A direction, which is worth considering and may result in getting the answer, is 
using the fact that a sufficient condition (for an access structures to 
be ideal) can be obtained by proving that the simple extension of a 
given uniform polymatroid is representable (cf. [6, Corollary 6.7]). 
This method has been applied in Section \ref{sec:ideal:acc:str} to the proof that all 
the structures described in Theorems \ref{thm:g2:zero} and \ref{thm:last:column} - 
\ref{thm:singleton} are ideal. 

It is worth noting that the class of access structures obtained from uniform 
polymatroids contains some interesting families of objects that can be useful
for the applications of secret sharing. The access structures discussed in Theorem 
\ref{thm:singleton} correspond to the organizational chart of an institution composed 
of several mutually independent departments managed by one superior unit. It follows 
from Theorem \ref{thm:singl:delta} that all those access structures are ideal.

Another interesting example is the family of uniform access structures characterized by 
Farr\`as et al. in \cite[Section VI]{FPXY} (cf. Remark \ref{rem:k:regular} below). It 
consists of multipartite access structures that are invariant under any permutation of 
blocks of participants. In other words all participants have the same rights, although 
they are not hierarchically equivalent.

A different situation occurs in compartmented access structures, where there is a set of 
distinguished participants, whose representatives must be present in all authorized 
sets. Such a case is described in Theorem \ref{thm:singl:delta}. 

This paper is intended to initiate research on the access structures obtained from 
uniform polymatroids, but it does not exhaust the topic and leaves space for further 
study. Some remarks on the new research possibilities can be found in Section 
\ref{sec:concl}. The Appendix contains a classification of all access structures with 
four parts obtained from uniform polymatroids.

%====================================================
\section{Preliminaries}
%====================================================

The aim of this section is to provide the necessary definitions and results regarding 
multipartite access structures and polymatroids.

Throughout the paper we use the following notations. The family of all subsets of a set 
$X$ is denoted by $\pow{X}$ (the power set). Similarly $\powk{k}{X}$ denotes the 
collection of all of $k$-element subsets of $X$. Let $\NN_0$ and $\NN$ denote the set of 
all non-negative integers and positive integers, respectively. Let $J$ be a finite set. 
For two vectors $\U=(u_x)_{x\in J},\V=(v_x)_{x\in J}\in\NN_0^J$ we write $\U\leq \V$ if 
$u_x\leq v_x$ for all $x\in J$. Moreover, $\U<\V$ denotes $\U\leq \V$ and $\U\neq \V$. 
Given a vector $\V=(v_x)_{x\in J}$, we define the support 
$\supp{\V}=\set{x\in J\  : \ v_x\neq 0}$ and the modulus $|\V|=\sum_{x\in J} v_x$. 
Furthermore, we write  $\V_X=(v'_x)_{x\in J}$, where $X\subseteq J$ and 
$$
v'_x=
\begin{cases}
v_x & \mbox{if } x\in X,\\
0   & \mbox{if } x\notin X.
\end{cases}
$$
In particular $\V_{\emptyset}=(0)_{x\in J}$. Let us observe that $|\V|=|\V_X|$ is 
equivalent to $\supp{\V}\subseteq X$. For every $z\in J$, we define the vector  
$\E{z}\in \NO^J$ such that $\E{z}=(e^{(z)}_x)_{x\in J}$ with $e^{(z)}_z=1$ 
and $e^{(z)}_x=0$ for all $x\neq z$.

%================================================================
\subsection{Multipartite access structures}\label{sec:multipart}
\label{ssec:multipartite}
%================================================================

Let $\Gamma$ be an access structure on a set of participants $P$. A participant $p\in P$ 
is said to be \emph{hierarchically superior or equivalent} to a participant $q\in P$ 
(written $q\preccurlyeq p$), if $A\cup \set{p}\in \Gamma$ for all subsets $A\subseteq 
P\setminus\set{p,q}$ with $A\cup\set{q}\in\Gamma.$ If $p\preccurlyeq q$ and 
$q\preccurlyeq p$, then the participants $p,q$ are called \emph{hierarchically 
equivalent}. Participants $p,q\in P$ are said to be \emph{hierarchically independent} if 
neither $p$ is hierarchically superior or equivalent to $q$ nor $q$ is  hierarchically 
superior or equivalent to $p$. 

By a \emph{partition} ($\Pi$-\emph{partition}) of the set of participants $P$ we mean a 
family $\Pi=(P_x)_{x\in J}$  of pairwise disjoint and nonempty subsets of $P$, called 
blocks such that $P=\bigcup_{x\in J}P_x$. An access structure $\Gamma$ is said to be 
\emph{multipartite} (\emph{$\Pi$-partite}) if all participants in every block $P_x$ are 
pairwise hierarchically equivalent. Thus we are allowed to define a hierarchy in $\Pi$. 
Namely, $P_x$ is said to be \emph{hierarchically superior or equivalent} to $P_y$ 
(written $P_y\preccurlyeq P_x$) if there are $p\in P_y$ and $q\in P_x$ such that 
$p\preccurlyeq q$. In other words it can be said that $P_y$ is \emph{hierarchically 
inferior or equivalent to $P_x$}. By transitivity we have $p \preccurlyeq q$ for all 
$p\in P_y$ and $q\in P_x$ whenever $P_y\preccurlyeq P_x$. The relation $\preccurlyeq$ 
both in $P$ and in $\Pi$ is reflexive and transitive but not antisymmetric in general, 
so it is a preorder. Moreover, this preorder is determined by the access structure $\G$, 
so it should be denoted by $\preccurlyeq_{\G}$. However, to simplify notation we write 
$\preccurlyeq$ if it does not lead to ambiguity. Similarly, blocks $P_x$ and $P_y$ are 
said to be \emph{hierarchically independent} if there are $q\in P_x$ and $p\in P_y$ such 
that $p$ and $q$ are hierarchically independent. On the other hand, if $P_x\preccurlyeq 
P_y$ or $P_y\preccurlyeq P_x$, then the blocks $P_x$ and $P_y$ are called 
\emph{hierarchically comparable}. Moreover, if $P_x\preccurlyeq P_y$ and 
$P_y\preccurlyeq P_x$, then the blocks $P_x$ and $P_y$ are called  \emph{hierarchically 
equivalent}. If $P_x\preccurlyeq P_y$ and the blocks are not hierarchically equivalent, 
then we write $P_x \prec P_y$.

Let us recall that a participant which does not belong to any minimal authorized set is 
called \emph{redundant}. It is easy to see that every participant is hierarchically 
superior or equivalent to any redundant participant. In particular, all redundant 
participants are hierarchically equivalent. A block of participants which contains a 
redundant participant will be also called \emph{redundant}.

A $\Pi$-partite access structure is said to be \emph{compartmented} if every pair of 
blocks in $\Pi$ is hierarchically independent. Otherwise the access structure is 
referred to as \emph{weakly hierarchical}. If an access structure is weakly hierarchical 
and no pair of blocks in $\Pi$ is hierarchically equivalent, then the access structure 
will be called \emph{hierarchical}. A hierarchical access structure such that every pair 
of blocks is hierarchically comparable  is referred to as \emph{totally hierarchical}. 
A complete characterization of ideal totally hierarchical access structure was presented 
by Farr\`as and Padr\'o \cite{FP}. It is worth pointing out that the phrase 
"compartmented access structure" used here is very general and covers several notions 
with the same name appearing in the literature.

Given a partition  $\Pi=(P_x)_{x\in J}$ of $P$ and a subset $A \seq P$ we define the 
vector $\pi(A)=(v_x)_{x\in J}$, where $v_x=|A\cap P_x|$.  If $\G$ is a $\Pi$-partite 
access structure, then all participants in every subset $P_x$ are pairwise 
hierarchically equivalent, so if $A\in \G, \ B\subseteq P$ and $\pi(A)=\pi(B)$, then 
$B\in \G$. We put $\pi(\G)=\set{\pi(A)\in \NO^J \ : \ A\in \G}$ and 
$$\pi(\pow{P})=\set{\pi(A) \in \NO^J\ :\ A\subseteq P}=\set{\V\in\NO^J\ : \ \V\leq \pi(P)}.$$
Obviously, if $A\subseteq B\subseteq P$, then $\pi(A)\leq \pi(B)$. Moreover, if 
$\U\in\pi(\G)$ and $\U\leq \V\leq \pi(P)$, then $\V\in \pi(\G)$. Indeed, there is $A\in 
\G$ such that $\U=\pi(A)$. The set $A$ can be extended to a set $B\subseteq P$ such that 
$\V=\pi(B)$. Hence $B\in \G$ and consequently $\V\in \pi(\G)$. This shows that 
$\pi(\G)\subseteq \pi(\pow{P})$ is a set of vectors monotone increasing with respect to 
$\leq$. On the other hand, every monotone increasing set $\G'\subseteq \pi(\pow{P})$ 
determines the $\Pi$-partite access structure $\G=\set{A\subseteq P \ : \ \pi(A)\in 
\G'}$. This shows that there is a one-to-one correspondence between the family of 
$\Pi$-partite access structures defined on $P$ and the family of monotone increasing 
subsets of $\pi(\pow{P})$. Therefore we use the same notation $\G$ for both the access 
structure and its vector representation.

The hierarchy among blocks in $\Pi$ can be characterized in vector terms as follows:
$P_y\preccurlyeq P_x$ if and only if 
\begin{equation}
\V-\E{y}+\E{x}\in \G \mbox{ for all } \V\in\G \mbox{ with } v_y\geq 1 \mbox{ and } v_x<|P_x|.
\label{eq:preceq}
\end{equation} 
To show that $P_y\preccurlyeq P_x$ it is enough to check if the above condition is 
satisfied for all vectors $v\in \min \Gamma$. A block $P_x$ in $\Pi$ is redundant if and 
only if $v_x=0$ for every $\V\in \min \G$.

%================================================================
\subsection{Polymatroids and access structures}
\label{sec:polymatroids}
%================================================================

Let $J$ be a nonempty finite set and let $\pow{J}$ denote the power set of $J$. \emph{A 
polymatroid} $\Pol$ is a pair $(J,h)$ where $h$ is a mapping 
$h:\pow{J} \longrightarrow \RR$ satisfying
\begin{enumerate}
	\item $h(\emptyset) = 0$;
	\item $h$ is monotone increasing: if $X \seq Y \seq J$, then $h(X)\leq h(Y)$;
	\item $h$ is submodular: if $X, Y \seq J$, then $h(X \cap Y) + h(X\cup Y)\leq h(X)+h(Y)$.
\end{enumerate}

The mapping $h$ is called \emph{the rank function} of a polymatroid. If all values of 
the rank function are integer, then the polymatroid is called \emph{integer}. An integer 
polymatroid $(J,h)$ such that $h(X)\leq |X|$ for all $X\seq J$ is called a 
\emph{matroid}. All polymatroids considered in this paper are assumed to be integer, so 
we will omit the term "integer" when dealing with integer polymatroid. 

Let $\Pol=(J,h)$ be a polymatroid and let $x\in J$ such that $h(\set{x})=1$.  The set 
$\set{X\in \pow{J\setminus \set{x}} \ : \  h(X\cup \set{x})=h(X)}$ is called a 
\emph{polymatroid port} or more precisely, the \emph{port of polymatroid $\Pol$ at the 
point $x$}. One can show that every polymatroid port is a monotone increasing family of 
some subsets of $J\setminus \set{x}$, which does not contain $\emptyset$.

The following examples of polymatroids play a special role in studying ideal access 
structures. Let $V$ be a vector space  of finite dimension and let 
$\mathcal{V}=(V_x)_{x\in J}$ be a family of subspaces of $V$. One can show that the 
mapping $h: \pow{J} \longrightarrow  \NO$ defined by $h(X) = \dim(\sum_{x\in X} V_x)$ 
for $X\in \pow{J}$ is the rank function of the polymatroid $\Pol = (J, h)$. The 
polymatroids that can be defined in this way are said to be \emph{representable}. If 
$\dim V_x\leq 1$ for all $x\in J$, then we obtain a matroid which is called 
\emph{representable} as well. The family $\mathcal{V}$ is referred to as a \emph{vector 
space representation} of the polymatroid (matroid). 
Let $\mathfrak{B}=(B_x)_{x\in J}$ be a family of finite sets. One can show that the mapping $h: \pow{J} \longrightarrow  \NO$ defined by 
$h(X) = |\bigcup_{i\in X} B_i|$ for $X\in \pow{J}$ is the rank function 
of the integer polymatroid $\Pol = (J, h)$. Every polymatroid that can be defined in 
this way is said to be \emph{Boolean} and the family $\mathfrak{B}$ is called the 
Boolean representation of the polymatroid. Boolean polymatroids are known to be 
representable.

The connection between matroids and ideal access structures was discovered by Brickell 
and Davenport \cite{BD}. They proved that if $\G \subseteq \pow{P}$ is the access 
structure of an ideal secret sharing scheme on a set of participants $P$ with a dealer 
$p_0\notin P$, then there is a matroid $\mathcal{S}$ with the ground set 
$P\cup\set{p_0}$ such that $\G$ is the port of $\mathcal{S}$ at the point $p_0$. This 
result can be stated as follows.

\begin{thm}[E.F. Brickell, D.M. Davenport \cite{BD}]
\label{thm:necess:cond}
Every ideal access structure is a matroid port.
\end{thm}

The converse is not true. For example, the ports of the Vamos matroid are not ideal 
access structures (cf. \cite{pS}). The following result is obtained as a consequence of 
the linear construction of ideal secret sharing schemes due to Brickell \cite{eB}.

\begin{thm}[E.F. Brickell \cite{eB}]\label{thm:suffic:cond}
Every port of a representable matroid is an ideal access structure.
\end{thm}

Let $\Pol=(J,h)$ be a polymatroid. For $J'=J\cup \set{x_0}$ with a certain $x_0\notin J$ 
and a monotone increasing family $\Delta\subseteq \pow{J}\setminus\set{\emptyset}$  
we define the function $h':\pow{J'}\longrightarrow \NO$ by 
$h'(X)=h(X)$ for all $X\in \pow{J}$ and 
$$h'(X\cup\set{x_0})=\begin{cases}
h(X) & \mbox{if } X\in \Delta,\\
h(X)+1 & \mbox{if }  X\in \pow{J}\setminus \Delta.
\end{cases}
$$
If $h'$ is monotone increasing and submodular, then $\Delta$ is said to be 
\emph{compatible} with $\Pol$ and $\Pol'=(J',h')$ is a polymatroid which is called the 
\emph{simple extension of $\Pol$ induced by $\Delta$}. It is easy to see that 
$h'(x_0)=1$ and $\Delta$ is the polymatroid port of $\Pol'$ at the point $x_0$. 
The next result, which is a consequence of \cite[Proposition 2.3]{lC} is very useful in 
the investigation of access structures induced by polymatroids.

\begin{lem}[\cite{lC} L. Csirmaz]
\label{lem:csirmaz}
A monotone increasing family $\Delta\seq \pow{J}\setminus \set{\emptyset}$ is compatible 
with an integer polymatroid $\Pol = (J,h)$ if and only if the following conditions are 
satisfied:
\begin{enumerate}[(1)]
	\item If $Y\subseteq X\subseteq J$ and $Y\notin \Delta$ while $X\in \Delta$, then 	
	$h(Y)<h(X)$. 
	\item If $X,Y\in \Delta$ and $X\cap Y\notin \Delta$, then 
	$h(X\cap Y)+h(X\cup Y) < h(X)+h(Y)$.
\end{enumerate}
\end{lem}

The following notation will be used very often throughout the paper.
Let $\Pol=(J,h)$ be a polymatroid and let $X\subseteq J$. We define the following set 
\begin{equation}\label{eq:B,Z,X}
\B(\Pol,X)=\set{\V\in \NO^{J}\ : \ \supp{\V}\subseteq X,\   |\V|=h(X),\ 
\forall_{Y\subseteq X} |\V_Y|\leq h(Y)}.
\end{equation}
It is easy to see that
\begin{equation}
\label{eq:equal_h}
\mbox{if } Y\seq X\seq J  \mbox{ and } h(Y)=h(X), \mbox{ then } \B(\Pol,Y)\seq \B(\Pol,X).
\end{equation}
On the other hand, $\B(\Pol,Y)\cap\B(\Pol,X)=\emptyset$ whenever $h(Y)\neq h(X)$.

The definition of $\B(\Pol,X)$ is related to the concept of bases of polymatroids. 
Given a polymatroid $\Pol=(J,h)$ and $\emptyset\neq X\subseteq J$, we consider $\Pol|X=(X,h|X)$, where $h|X\ :\ \pow{X}\longrightarrow \NO$. It is easy to see that $\Pol|X$ is a polymatroid which is called the \emph{restriction} of $\Pol$ to $X$.
Let $\Pol=(J,h)$ be a polymatroid. We define the set of bases of $\Pol$ by  
$\B(\Pol)=\set{\V\in \NO^J\ :\ |\V_X|\leq h(X) \mbox{ for all } X\subseteq J \mbox{ and } |\V|=h(J)}$. It is known that every polymatroid is uniquely determined by the set of its bases.
The set $\B(\Pol,X)$ is obtained from the set  $\B(\Pol|X)\subseteq \NO^X$ of bases of the restriction $\Pol$ to $X$ by the canonical embedding of $\NO^X$ into $\NO^J$. 

Now, we recall an important theorem of Farr\`as, Mart\'i-Farr\'e and Padr\'o \cite{FMFP} 
that characterizes those multipartite access structures that are matroid ports.

\begin{thm}{\cite[Theorem 5.3]{FMFP}}
\label{thm:mingamma}
Let $\Pi = (P_x)_{x\in J}$ be a partition of a set $P$ and let $\G$ be a connected 
$\Pi$-partite access structure on $P$. Consider $\Delta= \supp{\G}$. Then $\G$ is 
a matroid port if and only if there exists an integer polymatroid $\Pol=(J,h)$ with
$h(\{x\})\leq |P_x|$ for every $x\in J$ such that $\Delta$ is compatible with $\Pol$ 
and $\min\pi(\G) = \min\bigcup_{X\in \Delta} \B(\Pol,X)$.
\end{thm}

\begin{rem} \label{rem:represent}
\rm
Let $\Pi=(P_x)_{x\in J}$ be a partition of a set $P$. Let  
$\Delta\subseteq \pow{J}\setminus\set{\emptyset}$ be a monotone increasing family 
compatible with a polymatroid $\Pol=(J,h)$ such that $h(\set{x})\leq |P_x|$ for all 
$x \in J$.  Farr\`as, Mart\'i-Farr\'e, Padr\'o \cite{FMFP} proved that if the simple 
extension of $\Pol$ determined by $\Delta$ is a representable polymatroid, then the 
multipartite access structure $\G$ such that 
$\min \G=\min \bigcup_{X\in \Delta} \B(\Pol,X)$ is ideal. 
This result generalizes the result of Brickell \cite{eB}.
\end{rem}

Theorem \ref{thm:mingamma} can be used as a simple tool for constructing multipartite 
access structures which are matroids ports.

\begin{defn}
\label{defn:constr:acc:str}
\rm
For a given partition $\Pi=(P_x)_{x\in J}$ 
we take a polymatroid $\Pol=(J,h)$ with $h(\set{x})\leq |P_x|$ for every 
$x\in J$, a~monotone increasing family $\Delta\subseteq \pow{J}\setminus\set{\emptyset}$ 
which is compatible with $\Pol$ and we construct the smallest monotone increasing family 
$\G'\subseteq \pi(\pow{P})$ which contains $\min \bigcup_{X\in \Delta} \B(\Pol,X)$.
In other words, $\G'$ is the only monotone increasing family contained in $\pi(\pow{P})$ 
such that $\min\G'\seq \bigcup_{X\in \Delta} \B(\Pol,X)\seq \G'$. Obviously, 
$\G=\set{A\seq P\ :\ \pi(A)\in \G'}\seq \pow{P}$ is the access structure in the set of 
participants induced by its vector representation $\G'$. Both $\G$ and $\G'$ will be 
called the \emph{$\Pi$-partite access structure determined by a polymatroid $\Pol$ and a 
monotone increasing family $\Delta$} and will be denoted by $\G(\Pi,\Pol, \Delta)$.
\end{defn}

\begin{ex}
\rm
Let us consider $J'=\set{0,1,2,3}$ and the function
$h' : \pow{J'} \longrightarrow \NO$ defined by 
$$h'(X)=\begin{cases}
0 & \mbox{ if } |X|=0;\\
1 & \mbox{ if } |X|=1;\\
2 & \mbox{ if } |X|\geq 2.
\end{cases}
$$
It is easy to check that $\Pol'=(J',h')$ is a polymatroid and $\Delta=\set{\set{1,2}, \set{1,3}, \set{2,3},\set{1,2,3}}$ is its port at $0$. 
Moreover, $\Pol'$ is a simple extension of $\Pol=\Pol'|J$, where $J=\set{1,2,3}$. Thus 
$\Delta$ is compatible with $\Pol$. Hence we get $\B(\Pol,\set{1,2})=\set{(1,1,0)}$, 
$\B(\Pol,\set{1,3})=\set{(1,0,1)}$, $\B(\Pol,\set{2,3})=\set{(0,1,1)}$ and  
$\B(\Pol,\set{1,2,3})=\set{(1,1,0),(1,0,1),(0,1,1)}$. 
Now we are ready to define an access structure $\Gamma$ assuming $\min \Gamma= \min \bigcup_{X\in \Delta} \B(\Pol,X)= \set{(1,1,0),(1,0,1),(0,1,1)}$. It is easily seen 
that, a vector $\V$ is authorized in $\Gamma$ if and only if $|\supp{\V}|\geq 2$. 
\end{ex}

According to Theorem \ref{thm:mingamma} the access structure obtained in this way 
satisfies necessary condition to be ideal. The results of \cite{FMFP} mentioned in 
Remark \ref{rem:represent} provides a sufficient condition for $\G(\Pi,\Pol, \Delta)$ 
to be ideal.

\begin{rem}
\label{rem:nonconnectivity}
Assume that $h(\set{x})=0$ for a certain $x\in J$. Suppose $\V$ is a minimal vector in 
$\G(\Pi,\Pol,\Delta)$, then there is $X \in \Delta$ such that $\V\in \B(\Pol,X)$. By 
definition, $v_y\leq h(\set{y})$ for all $y\in J$. In particular $v_x\leq h(\set{x})=0$, 
i.e. no participant from $P_x$ belongs to $X$. This shows that if $h(\set{x})=0$, then 
all participants in $P_x$ are redundant, so every access structure induced by $\Pol$ is 
not connected. Therefore, from now on we assume that $h(\set{x})>0$ for all $x\in J$.
\end{rem}

%================================================================
\subsection{Uniform polymatroids}\label{sec:uniform}
%================================================================

We begin this subsection with the definition of uniform polymatroids which play a major 
role in this paper. To shorten notation we set $I_m=\set{0,1,\ldots,m}$.

\begin{defn}
An integer polymatroid $\Pol=(J,h)$ is called \emph{uniform} if
\[ |X|=|Y| \ \Longrightarrow\ h(X)=h(Y) \text{ \  \ for all } X, Y\subseteq J. \]
\end{defn}

Let $m:=|J|$. We define $h_i = h(X)$ for every $i=0,1,\ldots,m$ with $X\subseteq J$, 
$|X|=i$. It is obvious, that the sequence $(h_i)_{i\in I_m}$ determines the rank 
function of the polymatroid. For this sequence we define the \emph{increment sequence} 
$\bfg=(g_i)_{i\in I_m}$ by $g_i=h_{i+1}-h_i$ for $i=0,\ldots,m-1$ and additionally 
$g_m=0$. It is easy to see that $\bfg$ is nonincreasing sequence of non-negative 
integers. 

On the other hand, if  $\bfg=(g_i)_{i\in I_m}$, is a nonincreasing sequence of 
nonnegative integers with $g_m=0$, then we can define the sequence $(h_j)_{j\in I_m}$ by 
the formula 
\begin{equation} 
\label{eq:rankf}
h_j = \sum_{i=0}^{j-1} g_i \quad\text{ for all } j=1,\dots,m \text{ and }  h_0 = 0.
\end{equation}
Given a finite set $J$ with $|J|=m>0$, the numbers $h_j$ define a rank function 
$h \colon \pow{J} \to \NO$ of a uniform polymatroid $(J,h)$ by putting $h(X) = h_{|X|}$ 
for $X\subseteq J$. It is not difficult to notice that 
\begin{equation} 
\label{eq:sums}
h_k-h_j = \sum_{i = j}^{k-1} g_i \quad \text{  for all }  j,k \in I_m, \  j<k.
\end{equation} 
Notice also that $g_0 = 0 \Longleftrightarrow h_1 = \dots = h_{m} = 0$ \ and \ 
$g_1 = 0 \Longleftrightarrow h_1 = \dots = h_{m} = g_0$. Hence, according to the 
assumption that we consider only polymatroids such that their range functions do not 
have all values equal to 0, from now on we assume that for all sequences $\bfg$ and for 
all uniform polymatroids $\Pol$ we have $g_0 \neq 0$ or equivalently $h_1 \neq 0$.
To avoid repetition in the further part of the paper, a uniform polymatroid will be 
denoted by $\Pol=(J,h,\bfg)$ where $\bfg=(g_i)_{i\in I_m}, \ \  g_0>g_m=0$ is a 
nonincreasing sequence of nonnegative integers and $h: \pow{J}\to \NO$ is the rank 
function such that $h(X)=h_k=\sum_{i=0}^{k-1} g_i$ for every $X\in \pow{X}$ with $k=|X|$.

\begin{rem}
\label{rem:uniform:ideal}
We want to show that every uniform polymatroid determines an ideal access structure. 
Indeed, uniform polymatroids are known to be representable (cf. 
\cite[Theorem 6]{FMBPV}). Let $\KK$ be a finite field and let $(V_x)_{x\in J}$ be a 
$\KK$-vector space representation of a uniform polymatroid $\Pol=(J,h,\bfg)$. Then $V_x$ 
are subspaces of the vector space $\KK^{h_m}$ and $\dim V_x=h_1=g_0$ for every $x\in J$. 
For any $X\subseteq J$ we define $V_X=\sum_{x\in X} V_x$. 
Given a non zero vector $\beta\in \KK^{h_m}$, the family $\Delta=\set{X\subseteq J \ :\ \beta \in V_X} \subseteq \pow{J}$ is a monotone increasing family of subsets of $J$ and 
$\Delta$ is compatible with the polymatroid $\Pol$. It is easily seen that $(V_x)_{x\in J\cup\set{x_0}}$, where $x_0\notin J$ and $V_{x_0}=\Span(\set{\beta})$ is a vector space 
representation of the simple extension of $\Pol$ induced by $\Delta$. According to 
Remark \ref{rem:represent} the access structure $\Gamma(\Pi,\Pol,\Delta)$ is ideal. 

Varying the representation of $\Pol$ and the vector $\beta$, we can control to some 
extent the selection of $\Delta$ which allows us to obtain different ideal access  
structures. For example, if we take $\beta\in V_X$ for a certain $X\subset J$, then $X\in \Delta$. If $X\neq J$ and $\beta\notin V_X$, then $\Delta\cap \pow{X}=\emptyset$. 
More examples can be found in Section \ref{sec:ideal:acc:str}.
On the other hand, the characterization of those families $\Delta$ that cannot be 
obtained in this way does not seem to be an easy task.
\end{rem}

In order to continue our studies, we need to prove several technical lemmas. Let us 
recall that $\B(\Pol, X)$ is defined by Equation \eqref{eq:B,Z,X}.

\begin{lem}
\label{lem:min:coeffic}
Let $\Pol = (J, h, \bfg)$ be a uniform polymatroid. Assume that $X\subseteq J$, 
$1\leq k = |X|$ and ${\W\in \B(\Pol, X)}$. Then:
\begin{enumerate}[(1)]
\item For every $x\in X$ we have $w_x\geq g_{k-1}$.
\item If $w_x=g_{k-1}$ for some $x\in X$, then 
$\W-w_x\E{x}\in \B(\Pol,X\setminus\set{x})$.
\end{enumerate}
\end{lem}

\begin{proof}
(1) Let us notice that $|\W_X| = h(X) = h_k$ and $|\W_{X\setminus\set{x}}| \leq 
h(X\setminus\set{x}) = h_{k-1}$, hence
\[ w_x = |\W_X| - |\W_{X\setminus\set{x}}| \geq h_k - h_{k-1} = g_{k-1}.\]
(2) If we set $\V:=\W-w_x\E{x}$, then we have $\supp \V\subseteq X\setminus \set{x}$ and 
\[ |\V| = h_k - g_{k-1} = h_{k-1} = h(X\setminus \set{x}). \]  
\end{proof}

\begin{lem}
\label{lem:both:in:supp}
Let $\Pol=(J,h,\bfg)$ be a uniform polymatroid. Let $x,y\in X\seq J,\ x\neq y$ and 
$\W\in \B(\Pol,X)$ such that $w_x=g_0, \ w_y\neq 0$. 
If $\V\in \B(\Pol,\supp\V)$ and $\V\leq \W-\E{y}+\E{x}$, then $y\notin \supp \V$. 
\end{lem}

\begin{proof}
Let $\W':=\W-\E{y}+\E{x}$ and  $Y:=\supp\V$. It is clear that $\V\in \B(\Pol,Y)$ implies 
$v_x\leq h_1=g_0$ and $|\V|=h(Y)$. Moreover, $Y\seq X$ and $|\W_Y|\leq h(Y)$.  Suppose 
that $y\in Y$. If $x\in Y$, then we have 
$$h(Y)=|\V|\leq w_x+(w_y-1)+|\W'_{Y\setminus\set{x,y}}|=|\W_Y|-1\leq h(Y)-1,$$
a contradiction.

Similarly, if $x\notin Y$, then we have 
$$h(Y)=|\V|\leq (w_y-1)+|\W'_{Y\setminus\set{y}}|=|\W_Y|-1\leq h(Y)-1,$$
a contradiction.
This completes the proof.
\end{proof}

\begin{lem}
\label{lem:one:in:supp}
Let $\Pol=(J,h,\bfg)$ be a uniform polymatroid. Let $y\in X\seq J$ and 
$x\in J\setminus X$ and  $\W\in \B(\Pol,X)$ such that $w_y=g_0$. 
If $k:=|X|$, $g_k>0$ and $\V\in \B(\Pol,\supp\V)$ such that $\V\leq \W-\E{y}+\E{x}$, 
then $y\notin \supp{\V}$. Moreover, if $g_0>1$, then $x,y\notin \supp{\V}$, i.e. 
$\supp{\V}\seq X\setminus \set{y}$.
\end{lem}

\begin{proof}
Let $\W':=\W-\E{y}+\E{x}$. Clearly, $\supp{\V}\seq \supp{\W'}\seq X\cup\set{x}$. 
Let $Y:=X\cap \supp{\V}$ and let $l:=|Y|$. Suppose that $y\in \supp{\V}$. If 
$x\in \supp{\V}$, then $\supp{\V}=Y\cup\set{x}$ and we have $l\leq k$ and
$$
h_{l+1} = |\V|\leq |\W'_{Y\setminus\set{y}}|+1+(g_0-1)=|\W_{Y\setminus\set{y}}|+g_0=
|\W_{Y}|\leq h_{l}.
$$
Hence $0<g_{k}\leq g_{l}=h_{l+1}-h_{l}\leq 0$, a contradiction.

If $x\notin \supp{\V}$, then $\supp{\V}=Y$ and we have 
$$
h_{l} = |\V|\leq |\W'_{Y\setminus\set{y}}|+(g_0-1)=
|\W_{Y\setminus\set{y}}|+g_0-1=|\W_{Y}|-1=h_{l}-1,
$$ 
a contradiction. Thus we have proved that $\supp{\V}\seq (Y\setminus\set{y})\cup\set{x}$.
 
Now we assume $g_0>1$, and suppose $\supp{\V}=(Y\setminus\set{y})\cup\set{x}$. 
$$
h_{l} = |\V|\leq |\W'_{Y\setminus\set{y}}|+1=
|\W_{Y\setminus\set{y}}|+1=|\W_{Y\setminus\set{y}}|+g_0-(g_0-1)=|\W_{Y}|-(g_0-1)<h_{l},
$$
as $g_0-1>0$, a contradiction.
This shows $\supp{\V}=Y\setminus \set{y}\seq X\setminus \set{y}$ which completes the 
proof.
\end{proof}

\begin{lem}
\label{lem:onestep}
Let $\Pol = (J, h, \bfg)$ be a uniform polymatroid and let ${x,y\in X\seq J}$, $x\neq y$, as well as $\W\in \B(\Pol, X)$. If $w_y > 0$, then 
$\W':=\W-\E{y}+\E{x} \in \B(\Pol, X)$ or there exists a set $Y\subseteq X\setminus \{y\},\ x\in Y$, such that $\V:=\W_Y\in \B(\Pol, Y)$. Furthermore $\V\leq  \W$ and $\V\leq  \W'$.
\end{lem}

\begin{proof}
Note that $ \supp{\W'}\subseteq X $ and $ | \W'_X | = | \W_X | = h (X) = h_ {| X |} $. 
Let us consider the case $\W'\notin \B(\Pol,X)$, that is, there is a~set $Y \subseteq X$ 
that $|\W'_Y| \geq h(Y)+1$. Let us choose a minimum set $Y$ for this property. It is 
easy to see that $x\in Y$ and $y\notin Y$. Setting the notation $l: = |Y|$, we get 
$$h_l+1\leq |\W'_Y| = (w_x + 1) + |\W_{Y\setminus \set{x}}|=|\W_Y|+1 \leq h_l + 1, $$
and consequently $|\W_Y| = h_l$. Thus, for $\V:=\W_Y$ we have $\V\in\B(\Pol,Y)$. It is 
clear that $\V\leq\W$ and $\V\leq\W'$, which completes the proof.
\end{proof}

Now we introduce a notion of a vertex vector. Let $J$ be a finite set and $m:=|J|$ and 
let $\bfg=(g_i)_{i\in I_m}$ be the increment sequence of a uniform polymatroid 
$\Pol=(J,h,\bfg)$. Given $X\subseteq J$ and a bijection 
$\sigma: X\to \set{0,1,\ldots,k-1}$ where $k=|X|$, we define the vector 
$\W=(w_x)_{x\in J}$ by 
$$\W=\sum_{x\in X} g_{\sigma(x)}\E{x}$$
which is referred to as a \emph{vertex vector with basic set $X$}. Notice that in 
general we have $\supp{\W}\seq X$, but $\supp{\W}= X$ whenever $g_{k-1}>0$. Vertex 
vectors are the vertices of the convex polytope
$$
T =\set{\W\in \NO^J\ :\ |\W_X|\leq h(X) \mbox{ for every } X\seq J}
$$
determined by a polymatroid $(J,h)$.

\begin{lem}
\label{lem:vertex2}
Let $\Pol = (J, h,\bfg)$ be a uniform polymatroid. Then for every vertex vector  $\W$ we 
have $\W~\in~\B(\Pol,{\supp \W})$. 
\end{lem}

\begin{proof}
Let $\W$ be any vertex vector and $k:=|\supp{\W}|$. Let us take a subset 
$Y\subseteq \supp\W$ and set $l:=|Y|\leq k$. The sequence $\bfg$ being nonincreasing 
implies
$$ |\W_Y|=\sum_{x\in Y}w_x=\sum_{x\in Y}g_{\sigma(x)}\leq\sum_{i=0}^{l-1}g_i=h_l=h(Y).$$
Here we use the fact that the sum of $l$ arbitrary elements of a nonincreasing sequence 
does not exceed the sum of the $l$ initial entries of the sequence. 
In particular, if we get $|\W_{\supp \W}| = \sum_{i=0}^{k-1} g_i = h_k= h(\supp \W)$, 
which shows that $\W \in \B(\Pol, \supp\W)$.
\end{proof}

\begin{rem}
\label{rem:nonempty}
Notice that if $\Pol$ is a uniform polymatroid, then the set $\B(\Pol,X)$ is always 
nonempty since it contains vertex vectors with basic set $X$. In extreme cases when 
$X=\emptyset$ or the range function of the polymatroid  has all values equal to 0, the 
family $\B(\Pol,X)$ contains only the zero vector. Moreover, it is easy to check that if 
$\W\in \B(\Pol, X)$ for some $X\subseteq J$, then $\W\in \B(\Pol, \supp \W)$.
\end{rem}

Deciding if a monotone increasing family is compatible with given a polymatroid is not 
easy task. The Csirmaz Lemma seems to be the most general tool for solving this problem. 
For example, it is easy to check, that if the increment sequence of a polymatroid with 
ground set $J$ is strictly decreasing, then every proper monotone increasing family of 
subsets of $J$ is compatible with the polymatroid. 
At the end of this section we present several facts related to the compatibility 
of monotone increasing families and polymatroids which are used in proofs in subsequent 
sections.

\begin{lem}
\label{lem:head} 
Let $\Pol = (J, h, \bfg)$ be a uniform polymatroid and let a monotone increasing family 
$\Delta\subseteq \pow{J}\setminus\set{\emptyset}$ be compatible with $\Pol$.
\begin{enumerate}[(1)]
	\item If $g_{k} = 0$ for some $1 \leq k \leq |J|$, then all subsets of the set $J$ 
	with at least $k$ elements  belong to $\Delta$.
	\item If $\Delta$ contains a minimal set with $k$ elements, then $g_{k-1}>0$. 
\end{enumerate}
\end{lem}

\begin{proof}
(1) By assumption we have $g_i=0$ for all $i=k,\ldots,m$. Let us consider 
$X \subseteq J$, $l:=|X| \geq k$. Then we have 
$$h(J) - h(X) = h_{|J|} - h_{|X|} = \sum_{i=l}^{m-1} g_i =0.$$
This implies $h(X)=h(J)$ and by the Csirmaz Lemma we get $X\in \Delta$.

(2) Assume that $X\seq J$ is a minimal set in $\Delta$, $|X|=k$. Then for every 
$Y\subseteq X$ with $|Y|= k-1$ we have $Y\notin \Delta$, so by the Csirmaz Lemma 
$h_{|Y|}<h_{|X|}$. Hence $$g_{k-1}=h_k-h_{k-1}=h_{|X|}-h_{|Y|}>0.$$ 
\end{proof}

\begin{lem}
\label{lem:compatible}
If $\Delta\subseteq \pow{J}\setminus\set{\emptyset}$ is a monotone increasing family 
such that $\min\Delta = \{X\}$ for some $\emptyset\neq X\subseteq J$, then $\Delta$ is 
compatible with a uniform polymatroid $\Pol = (J, h,\bfg)$ if and only if $g_{m-1} > 0$.
\end{lem}

\begin{proof}
Assume $\Delta$ is compatible with $\Pol$. If $x\in X$, then $J\setminus\set{x}\notin \Delta$, so by Csirmaz Lemma $h(J\setminus \set{x})<h(J)$, thus 
$g_{m-1}=h(J)-h(J\setminus \set{x})>0$. 

Now we shall show that the conditions of the Csirmaz Lemma are met whenever $g_{m-1}>0$. 
Let us notice that $h_i-h_{i-1}=g_{i-1}>0$ for all $i=1,\ldots,m$, so the sequence 
$h_0,h_1,\ldots,h_m$ is strictly increasing. Let us take such sets $Y \subseteq W \subseteq J$, that $Y \notin \Delta$ and $W \in \Delta$. Of course, $|Y| < |W|$, so we 
have $h(Y)<h(W)$, thus the condition (1) is satisfied.

Now let us consider $W,Y\in\Delta$. Then $X\subseteq W$ and $X \subseteq Y$ since 
$\min\Delta = \{X\}$, so $W\cap Y\in \Delta$. This shows that the second condition of 
the Csirmaz Lemma is also satisfied.
\end{proof}

Let us recall a result of Farr\`as, Padr\'o, Xing and Yang, which can be restated as 
follows.

\begin{lem}[\cite{FPXY}, Lemma 6.1]
\label{lem:k-compatible}
For a positive integer $k\in I_m$, the monotone increasing family $\Delta$ such that 
$\min \Delta=\powk{k}{J}$ is compatible with a uniform polymatroid $\Pol=(J,h,\bfg)$ if 
and only if $g_{k-1}>g_k$.
\end{lem}

Further results concerned with compatibility can be found in Section 
\ref{sec:hier:pord}.

%================================================================
\section{Access structures determined by uniform polymatroids}  
\label{sec:acc:str}
%================================================================

This section is devoted to the study of those uniform polymatroids that determine weakly 
hierarchical access structures.From now on we make the assumption that $J$ is a finite 
set with $m:=|J|\geq 2$. A partition of a set of participants is denoted by 
$\Pi=(P_x)_{x\in J}$. A uniform polymatroid is a triplet $\Pol=(J,h,\bfg)$ where $h: \pow{J}\longrightarrow \NO$ is the rank function and $\bfg=(g_i)_{i\in I_m}$ is the 
increment sequence of the polymatroid. Recall the sequence $\bfg$ is nonincreasing and 
$g_m=0$. Moreover, we assume $0<g_0<|P_x|$ for all $x\in J$.  Next we consider a 
monotone increasing family $\Delta\seq \pow{J}\setminus\set{\emptyset}$ that is 
compatible with $\Pol$. Finally, $\G(\Pi,\Pol,\Delta)$ is the access structure 
determined by $\Pi,\ \Pol,\ \Delta$ as defined in Definition \ref{defn:constr:acc:str}. 
The relation $\preccurlyeq$ is the hierarchical preorder induced by $\Gamma$ in $\Pi$. 
We define for further use the following notations 
$\eta(\bfg)=\min\set{i\in I_m\ :\ g_i=0}$ and $\mu(\Delta)=\min\set{|X|\ :\ X\in 
\Delta}$. The above settings ensure that $\eta(\bfg)\geq 1$ and $\mu(\Delta)\geq 1$.

\begin{ex} 
\label{ex:threshold}
Let us consider a uniform polymatroid $\Pol=(J,h,\bfg)$ such that $\eta(\bfg)=1$, i.e.  
$g_0>g_1=0$ and a monotone increasing family $\Delta$ of subsets of $J$ compatible with 
$\Pol$. Applying Lemma \ref{lem:head} (1) yields 
$\Delta=\pow{J}\setminus \set{\emptyset}$. According to Equation \eqref{eq:rankf} we 
have $h(X)=g_0$ for all nonempty subsets $X$ of $J$. Hence $\B(\Pol,X)\seq \B(\Pol,J)$ 
for every $\emptyset\neq X\seq J$ (cf. Equation \eqref{eq:equal_h}) and consequently 
$\bigcup_{X\in \Delta} \B(\Pol,X)=\B(\Pol,J)$. Let $\GPPD$. This implies that 
$\W\in\min\Gamma$ if and only if $|\W|=g_0$ or equivalently $\W\in \Gamma$ if and only 
if $|\W|\geq g_0$. This shows that the threshold access structure is the only type of 
access structures determined by uniform polymatroids with $\eta(\bfg)=1$. In particular 
all blocks (and participants) are hierarchically equivalent.
\end{ex}

Let us collect several simple observations, which are very helpful in many proofs.

\begin{lem}
\label{lem:supp:property}
For $\G=\G(\Pi,\Pol,\Delta)$ we have
\begin{enumerate}[(1)]
\item $ \B(\Pol, X)\subseteq  \G$ for all  $X \in \Delta$.
\item $\supp \G = \Delta$. 
\item If $\W \in \Gm$, then $ \W \in \B(\Pol, \supp \W)$ and $\supp\W\in\Delta$.
\item If  $\W \in \G$, then there exists $\V \in \Gm$ such that $\V \leq  \W$, $ \V \in 
\B(\Pol, \supp \V)$  and $\supp \V \in \Delta$.
\item If $\W$ is a vertex vector and $\supp{\W}\in \Delta$, then $\W \in \G.$
\end{enumerate}
\end{lem}

\begin{proof}
(1) This follows directly from Definition \ref{defn:constr:acc:str}. 

(2) Let us consider $Y \in \supp \G$. Then there exists $\W \in \G$ such that 
$\supp \W = Y$. Let us consider two cases:

(i) $\W \in \Gm$. Then there exists $X \in \Delta$ such that $\W \in \B(\Pol, X)$, so  
$Y \subseteq X$.  If $Y=X$, then $Y \in \Delta$. If $Y \subsetneq X$, then also  
$Y \in \Delta$. Indeed, let us  notice that $|\W_Y|\leq h(Y)$, $|\W_X|= h(X)$ and 
$|\W_Y|=|\W_X|$, where the later equality follows from the fact 
$\supp \W = Y \subseteq X$. Moreover, if  $Y \not \in \Delta$, then by the Csirmaz Lemma 
we would get 
$$h(X)=|\W_X|=|\W_Y|\leq h(Y) < h(X),$$ 
which is a contradiction.

(ii) $\W \in \G$ and $\W \not \in \Gm$. Then there is $\V \in \Gm$ such that 
$\V\leq  \W$. From the case (i) we get $\supp \V \in \Delta$. Let us notice that 
$\supp \V \subseteq \supp \W$. Moreover, $\Delta$ is a monotone increasing family, so 
$Y=\supp \W \in\Delta$.

Now we shall show the converse inclusion. Let us take $X \in \Delta$. As we already have 
observed in Remark \ref{rem:nonempty}, the family $\B(\Pol, X)$ cannot be empty, so 
there is a certain vector $\W \in  \B(\Pol, X)$. By (1) we get $\W \in \G$, so 
$\supp \W \in \supp \G$. The family $\supp \G$ is monotone increasing and 
$\supp \W \subseteq X $, so $X \in \supp \G$.

(3) If $\W \in \Gm$, then $\W \in  \B(\Pol, X)$ for a certain $X \in \Delta$. Remark 
\ref{rem:nonempty} implies $\W \in \B(\Pol, \supp \W)$. Moreover, 
$\supp \W \in \supp \G$, hence and by (2) we get $\supp \W  \in \Delta$.

(4) It follows from (3) immediately.

(5) If $\W$ is a vertex vector, then we have $\W\in\B(\Pol, \supp \W)$ by Lemma 
\ref{lem:vertex2}. By assumption and part (1) of this lemma we get $\W\in \Gamma$.
\end{proof}

\begin{lem}
\label{lem:one:minimal}
Let $\GPPD$.
If $g_1=g_{n-1}>0$ for some $2\leq n\leq m$ and if $X,Y\in \min \Delta$ as well as 
$|X\cup Y|\leq n$, then $X=Y$ or both sets are singletons. Moreover, if $g_0=g_1$, then 
$X=Y$ even if both $X, Y$ are singletons.
\end{lem}

\begin{proof}
For $n = 2$, the claim is obvious. Let us assume $n \geq 3$. It is enough to consider 
the case $X \neq Y$. Suppose that at least one of these sets, for example $X$, has at 
least 2 elements. Let us fix $x \in X$ and consider the set
$$Y'=
\begin {cases}
Y & \mbox{when } X \cap Y \neq \emptyset; \\
Y \cup \{x \} & \mbox{when } X \cap Y = \emptyset.
\end {cases}
$$
Note that $|X \cup Y'| = |X \cup Y| \leq n$ and $W: = X \cap Y' \neq \emptyset$. In 
addition, $W$  is a proper subset of $X$ which is a minimum set in $\Delta$, so it does 
not belong to $\Delta$. Hence according to the Csirmaz Lemma we get 
$$h (W) + h (X \cup Y ') <h (X) + h (Y').$$
On the other hand, the assumption $g_1=g_{n-1}$ implies $h_l = g_0 + (l-1)g_1$ for every 
$1 \leq l \leq n$. From this we get
\[h_ {| W |} + h _ {(| X | + | Y '| - {| W |})} <h_ {| X |} + h_ {| Y' |}, \]
\[g_0 + (| W | - 1) g_1 + g_0 + (| X | + | Y '| - {| W |} - 1) g_1 <g_0 + (| X | - 1) g_1 + g_0 + (| Y '| - 1) g_1. \]
It is easy to see that the simplified expression above is $0 <0$, which gives a 
contradiction. This shows that if $X$ and $Y$ are different, then they cannot have more 
than one element.

Let us assume $g_0=g_1$ and $|X|=|Y|=1$. Let us suppose $X\neq Y$. Then 
$X\cap Y=\emptyset$, so by the Csirmaz Lemma we have 
$$ h(X\cap Y)+h(X\cup Y)<h(X) + h(Y)$$
and consequently $h_2<2h_1$ or equivalently $g_0+g_1<2g_0$, which is a contradiction.
\end{proof}

\begin{prop}
\label{prop:not:comparable}
If $X \in \min \Delta$, then for all $x, y \in X$, $x\neq y$, the blocks  $P_x$ and 
$P_y$ are hierarchically independent in the access structure $\GPPD$.  
\end{prop}

\begin{proof}
Let $X\in \min \Delta$ and let $x, y$ be two different elements in $X$. Suppose 
$P_y\preccurlyeq P_x$ and consider a vertex vector $\W$ 
with basic set $X$ and $w_x=g_0$. Setting $k:=|X|$ and applying Lemma \ref{lem:head} (2) 
we have $g_{k-1}>0$ so $\supp \W = X$, in particular $w_y>0$ and by Lemma 
\ref{lem:supp:property} (5) we get $\W\in \G$. Thus $\W'=\W-\E{y}+\E{x}\in \G$. By Lemma 
\ref{lem:supp:property} (4) there is $\V\in \min\G$ such that $\V\leq \W'$ and 
$\V\in \B(\Pol,\supp{\V})\seq \G$, so applying Lemma \ref{lem:both:in:supp} we have 
$y\notin \supp{\V}\subsetneq X$, which contradicts the fact that $X\in \min\Delta$. By a 
similar argument we show that the case $P_x\preccurlyeq  P_y$ is impossible.
\end{proof}

\begin{prop}
\label{prop:not:inferior}
If $X\in\min \Delta$, $1 \leq k:=|X|\leq m-1$ and $g_k>0$, then for every $y\in X$ the 
block $P_y$ is not hierarchically inferior or equivalent to any block $P_x\neq P_y$ in 
the access structure $\GPPD$.
\end{prop}

\begin{proof}
Let $y\in X$ and let us suppose that $P_y\preccurlyeq  P_x$ for a certain $x\in J$. 
By Proposition \ref{prop:not:comparable} we have $x\in J\setminus X$. Let us consider a 
vertex vector $\W$ with basic set $X$ and $w_y = g_0$. Obviously, $\W\in \G$ by Lemma 
\ref{lem:supp:property} (5). Then the vector $\W':=\W-\E{y}+\E{x}$ also belongs to $\G$. 

By Lemma \ref{lem:supp:property} (4) there exists a minimal authorized vector $\V$ such 
that $\V \leq  \W'$, $ \V \in \B(\Pol, \supp \V)$ and $\supp \V\in\Delta$. If $g_0>1$, 
then Lemma \ref{lem:one:in:supp} implies $\supp{\V}\seq X\setminus\set{y}$, but this 
contradicts the assumption $X\in \min \Delta$. 

If $g_0=1$, then $g_0=g_1=g_k$ and by Lemma \ref{lem:one:in:supp} we have 
$\supp{\V}\seq (X\setminus \set{y})\cup\set{x}$. For $Y\in \min\Delta$ such that 
$Y\seq \supp{\V}$ we have $X\cup Y\seq X\cup\set{x}$, so $|X\cup Y|\leq k+1$. Applying 
Lemma \ref{lem:one:minimal} yields $X=Y$ but $y\in X$ and $y\notin Y$, a contradiction.  
\end{proof}

\begin{lem}
\label{lem:technical}
Let $\bfg =(g_i)_{i\in I_m}$ be the increment sequence of a uniform polymatroid $\Pol$ 
and let $\GPPD$. Let us assume that $X\in \min \Delta$ with $k:=|X|$ and there are 
$x,y\in J, \ \ x\neq y$ such that $|X\cup \set{x,y}|\geq 3$ and the blocks 
$P_y$ and  $P_x$ are hierarchically comparable  in the access structure $\G$. 
Furthermore, we assume that $g_1=g_k$ and $g_l>0$ for a certain $1\leq l<m$. If 
$X\cap\set{x,y}\neq \emptyset$ or $g_0=g_1$, then $g_1= g_l$.
\end{lem}

\begin{proof} 
If $g_1=1$, then the claim is obvious.

Assume that $g_1>1$ and assume that this is not the case.
Let $l$ be the least positive integer that does not satisfy the claim. That means, 
$g_1=g_{l-1}>g_l>0$. Obviously, $k+1 \leq l \leq m-1$. This implies $k\leq m-2$. Without 
loss of generality we can assume  that $P_y\preccurlyeq P_x$. By Proposition  
\ref{prop:not:inferior} we have $y\notin X$. Let now  $Y \subseteq J$ be a set with 
$l+1$ elements which contains $X\cup \set{x,y}$. Moreover, let us take an element $z \in 
X\setminus \set{x,y}$. 

Let us consider a vertex vector $\W$ with basic set $Y$ and $w_x=g_0$ and $w_z = g_l$. 
Obviously, $\supp{\W}=Y$, as $g_l>0$. Under the above assumptions, $w_t=g_1$ for all 
$t\in Y\setminus\set{x,z}$, in particular we have $w_y = g_1$. For every $0<j\leq l$ 
we have 
\begin{equation}
\label{eq:hz1} 
h_j = \sum_{i=0}^{j-1}g_i = g_0 + (j-1)g_1.
\end{equation}
Let us notice that $Y \in \Delta$ as $X\subseteq Y$. Hence  $\W\in \G$ by Lemma 
\ref{lem:supp:property} (5). Moreover, $h_{l+1} = |\W| = g_0 + (l-1)g_1 + g_l$. 
Since $P_y\preccurlyeq P_x$ we have $\W':=\W-\E{y}+\E{x}\in \G$. Let us notice 
$\supp{\W'} = Y$. Now by Lemma \ref{lem:supp:property} (4) there exists a minimal 
authorized vector $\V$ such that ${\V\leq  \W'}$, $\V\in \B(\Pol, \supp \V)$ and 
$W:=\supp \V \in \Delta$. Lemma \ref{lem:both:in:supp} implies $y\notin W$, i.e. 
$W\subseteq Y\setminus \set{y}$, so $|W|\leq l$. Let $Z\in \min \Delta$ such that 
$Z\subseteq W$.

Thus $X\cup Z\seq X\cup W\seq Y\setminus\set{y}$, so $|X\cup Z|\leq l$ and applying 
Lemma \ref{lem:one:minimal} yields $X=Z$ or both $X$ and $W$ are singletons. If 
$X\cap\set{x,y}\neq \emptyset$, then $x, z\in X$, so $X$ is not a singleton, thus $X=Z$. 
If $g_0=g_1$, then $X=Z$. Thus in both cases, we have $z\in X=Z\seq W$, so applying 
Lemma \ref{lem:min:coeffic} (1) we get $g_{|W|-1 }\leq v_z$. Notice also that 
$v_z \leq w'_z = w_z = g_l<g_{l-1}\leq g_{|W|-1}$, a contradiction which proves that 
$g_1=g_l$. 
\end{proof}

\begin{prop}
\label{prop:two:offmin}
Let $\bfg =(g_i)_{i\in I_m}$ be the increment sequence of a uniform polymatroid $\Pol$ 
and let $\GPPD$. Let us assume  $n:=\eta(\bfg)\geq 3$. If there are $X\in \min \Delta$ 
such that $1\leq |X|\leq n-2$ and $x, y \in J \setminus X$ such that the blocks $P_x$ 
and $P_y$ are hierarchically comparable  in the access structure $\G$, then 
$g_0=g_1=\dots=g_{n-1}>g_n=0$.
\end{prop}

\begin{proof} 
If $g_0=1$, then let us observe that 
\[1 = h_1 = g_0 \geq g_1 \geq \dots \geq g_{n-1} \geq 1.\]
Hence $g_0=g_1 = \dots = g_{n-1}>g_n=0$.

Thus we assume $g_0\geq 2$. If the blocks $P_x$ and $P_y$ are hierarchically comparable, 
then one can assume without loss of generality that $P_y\preccurlyeq P_x$. Let us 
consider a vertex vector $\W$ with basic set $X\cup\set{y}$ such that $w_y = g_0$. 
Obviously, by Lemma \ref{lem:supp:property} (5) we have $\W\in \G$. Then the vector 
$\W':=\W-\E{y}+\E{x}$ belongs to $\G$. 
By Lemma \ref{lem:supp:property} (4) there exists a minimal authorized vector $\V$ such 
that $\V \leq \W'$, $\V\in\B(\Pol, \supp \V)$ and $\supp \V \in \Delta$. 
By Lemma \ref{lem:one:in:supp} we have $\supp{\V}\seq X$, but $X$ is minimal in 
$\Delta$, so  $\supp{\V}= X$. Thus we have 
\[h_k = |\V| \leq |\W'_X| = |\W_X| = \sum_{i=1}^k g_i=  h_{k+1} -g_0,\]
so $g_0 \leq h_{k+1} -h_k = g_k$. The sequence $\bfg$ is nonincreasing, so  
$g_0=g_1 = \dots = g_k$. Thus we have shown that $g_1 = \ldots = g_k$. To complete the 
proof it is enough to apply Lemma \ref{lem:technical}, assuming $l=n-1$. 
\end{proof}

\begin{prop}
\label{prop:one:offmin}
Let $\bfg =(g_i)_{i\in I_m}$ be the increment sequence of a uniform polymatroid $\Pol$ 
and let $\GPPD$. Let us assume  $n:=\eta(\bfg)\geq 3$. If there are $X\in \min \Delta$ 
with $2\leq |X|\leq n-1$ and $x\in X$ and $y \in J \setminus X$ such that the blocks 
$P_x$ and $P_y$ are hierarchically comparable  in the access structure $\G$, then $g_1 = 
\dots = g_{n-1}>g_n=0$.
\end{prop}

\begin{proof}
If the blocks $P_x$ and $P_y$ are hierarchically comparable, then it follows from 
Proposition \ref{prop:not:inferior} that $P_y\preccurlyeq P_x$. Let us consider a vertex 
vector $\W$ with basic set $X \cup \set{y}$ such that $w_x = g_0$ and $w_y = g_1$. 
Obviously, $\W\in \G$ by Lemma \ref{lem:supp:property} (5).  Then also the vector 
$\W':=\W-\E{y}+\E{x}$ belongs to $\G$ and ${\supp{\W'} \seq X\cup\set{y}}$. Hence by 
Lemma \ref{lem:supp:property} (4) there is a minimal authorized vector $\V$, such that 
$\V \leq \W'$, $\V \in \B(\Pol, \supp \V)$ and $Y:=\supp \V \in \Delta$. Let us observe 
$Y  \subseteq \supp{\W'} \seq X \cup\set{y}$. 

By Lemma \ref{lem:both:in:supp} we have $y\notin Y$ that shows $Y\seq X$, but $X$ is 
minimal in $\Delta$, so $Y=X$. Thus we have 
$$h_k=|\V|\leq g_0+\sum_{i=2}^k g_i=h_{k+1}-g_1$$
where $k:=|X|$. Hence $g_1\leq h_{k+1}-h_k=g_k$ and $g_1=g_k$ as the sequence $\bfg$ is 
nonincreasing. To complete the proof it is enough to apply Lemma \ref{lem:technical}, 
assuming $l=n-1$.
\end{proof}

\begin{cor}
\label{cor:n=m}
If there are $x,y\in J$ such that $P_x$ and $P_y$ are hierarchically comparable  and 
$3\leq |X\cup\set{x,y}|\leq\eta(\bfg)$ for a certain $X\in\min\Delta$, then 
$g_1=g_{m-1}$. Moreover, if $X\cap \set{x,y}=\emptyset$, then $g_0 = \dots =g_{m-1}$.
\end{cor}

\begin{proof}
Let us write $n:=\eta(\bfg)$. Assuming with no loss of generality that 
$P_y\preccurlyeq P_x$ we obtain that $\set{x,y}$ is not contained in $X$, by Proposition 
\ref{prop:not:comparable}, so $|X|\leq n-1$. Applying Proposition 
\ref{prop:not:inferior} yields $y\notin X$. If $x\in X$, then $2\leq |X|\leq n-1$ and 
applying Proposition \ref{prop:one:offmin} yields $g_1=g_{n-1}>g_n=0$. If $x\notin X$, 
in particular $|X|=1$, then applying Proposition \ref{prop:two:offmin} yields 
$g_0=g_1=g_{n-1}>g_n=0$.

Suppose, contrary to our claim that $n<m$. Then there is a subset $Z\seq J$ such that 
$|Z|=n+1$ and $X\cup \set{x,y}\seq Z$. Let us choose $z\in X\setminus\set{x,y}$ and 
denote $Z'=Z\setminus \set{z}$. Lemma \ref{lem:head} (1) implies that the set $Z'$ 
belongs to $\Delta$ but it is not minimal as $x,y\in Z'$. 
So there is $Y\in \min\Delta$ such that $Y\subsetneq Z'$. Applying again Proposition 
\ref{prop:not:inferior} we get $y\notin Y$, so $X\cup Y\seq Z\setminus \set{y}$, thus 
$|X\cup Y|\leq n$. If $|X|>2$, then we can apply Lemma \ref{lem:one:minimal} to get 
$X=Y\seq Z'$, which is a contradiction as $z\in X$ but $z\notin Z'$. If $|X|=1$, then  
$X\cap\set{x,y}=\emptyset$ and by Proposition \ref{prop:two:offmin} we have $g_0=g_1$ 
and applying again Lemma \ref{lem:one:minimal} yields  $X=Y\seq Z'$, a contradiction as 
before. This completes the proof.
\end{proof}

\begin{cor}
\label{cor:equiv}
Any multipartite access structures determined by uniform polymatroid $\Pol=(J,h,\bfg)$ does not admit hierarchically equivalent blocks unless $\eta(\bfg)=1$ or 
$g_0 = \dots =g_{m-1}$.
\end{cor}

\begin{proof}
It is shown in Example \ref{ex:threshold}, that all blocks are hierarchically equivalent 
in any access structure determined by a uniform polymatroid with $\eta(\bfg)=1$. Let 
$n:=\eta(\bfg)\geq 2$ and suppose that there are $x,y \in J$ such that $P_x$ and $P_y$ 
are hierarchically equivalent. Let us consider a subset $X\seq J$ such that $x,y\in X$ 
and $|X|=n$. Lemma \ref{lem:head} (1) and Proposition \ref{prop:not:comparable} imply 
that $X\in \Delta$ but $X$ is not minimal, so there is $Y\in\min\Delta$ such that 
$Y\subsetneq X$. By Proposition \ref{prop:not:inferior} $x,y\notin Y$. If $n=2$, then 
$Y=\emptyset$, which is a contradiction. Hence we get $n\geq 3$ and $3\leq |Y\cup 
\set{x,y}|\leq n$ and applying Corollary \ref{cor:n=m} yields $g_0=g_{m-1}$. 
\end{proof}

%================================================================
\section{Hierarchical preorder determined by access structure}  
\label{sec:hier:pord}
%================================================================

In this section we shall present several results on hierarchical orders induced by 
access structures determined by uniform polymatroids with ground set $J$ and monotone 
increasing families of subsets of $J$. It is worth noticing that such construction of an 
access structure is only possible if the monotone increasing  family is compatible with 
the given polymatroid. Let us recall that an access structure is called compartmented 
if every two different blocks are hierarchically independent. If there is at least one 
pair of blocks of participants which are hierarchically comparable, then the access 
structure is called weakly hierarchical. If weakly hierarchical access structure does 
not admit hierarchically equivalent blocks, then it is referred to as hierarchical. 

\begin{thm}
\label{thm:connected}
Let $\GPPD$ and let $\bfg =(g_i)_{i\in I_m}$ be the increment sequence of a uniform 
polymatroid $\Pol$. If $g_0>g_{m-1}$, then the access structure $\GPPD$ is connected.
\end{thm}

\begin{proof}
Given $x\in J$, we want to show that there is $\W\in\min \G$ such that $w_x\neq 0$. If 
there is $X\in \min\Delta, \ x\in X$ and $i:=|X|$, then $g_{i-1}>0$ by Lemma 
\ref{lem:head} (2). It is easy to see that any vertex vector $\W$ with basic set $X$  
belongs to $\min\Gamma$ and $w_x\neq 0$. Now we assume that $x\notin X$ for all 
$X\in \min\Delta$. Let us denote $l:=\min \set{i\in I_m\ :\ g_0>g_i}$. By assumption 
$l\leq m-1$. Let us take $X\in\min\Delta$ such that $k:=|X|=\mu(\Delta)$ and consider 
$Y\seq J$ such that $|Y|=\max\set{k,l}+1$ and $\set{x}\cup X\seq Y$. Let 
$\W$ be a vertex vector with basic set $Y$ such that $w_x=g_0$ and $w_y=g_{l}$ for a 
certain $y \in X$. Lemma \ref{lem:supp:property} (5) implies that $\W\in \Gamma$ as 
$Y\in \Delta$. Thus there is $\V\in\min\Gamma$ such that $\V\leq \W$. 
Since $\supp{\V}\in\Delta$, there is $W\in \min\Delta$ such that $W\seq\supp{\V}$. By 
assumption $x\notin W$, so $W\seq Y\setminus \set{x}$. 

It turns out that $W=X$. Indeed, if $k\geq l$, then $Y\setminus\set{x}=X$ and by the 
minimality of $X$ in $\Delta$ we have $W=X$. For the case $k < l$ we have $l\geq 2$, 
$g_0=g_1$ and $X\cup W\seq Y\setminus \set{x}$, thus $|X\cup W|\leq l$ and consequently 
$X=W$ by Lemma \ref{lem:one:minimal}. 

If $g_l=0$, then $v_y\leq w_y=g_l=0$, i.e. $y\notin \supp{\V}$ which contradicts the 
fact that $y\in X= W\subseteq \supp{\V}$.

If $g_l\neq 0$ and $v_x\neq 0$, then we have the claim. Let us suppose $v_x=0$, i.e. 
$x\notin \supp{\V}$. Thus for $Z=\supp{\V}$ we have
\begin{eqnarray*}
h(Z) & = & |\V_Z| = \sum_{z\in Z} v_z\leq\sum_{z\in Z}w_z=
w_y+\sum_{z\in Z\setminus\set{y}} w_z=
\sum_{z\in (Z\cup \set{x})\setminus\set{y}} w_z - (w_x-w_y)= \\
    & = & |\W_{(Z\cup \set{x})\setminus\set{y}}|-(w_x-w_y)\leq 
    h((Z\cup \set{x}) \setminus\set{y})-(g_0-g_l)<h(Z),
\end{eqnarray*}
a contradiction as $|Z|=|(Z\cup \set{x})\setminus\set{y}|$ and $g_0-g_l>0$.
\end{proof}

This theorem shows that the access structures determined by uniform polymatroids are 
connected  except for the ones in the column F of Table \ref{tab:general}.

\begin{thm}
\label{thm:extreme:columns}
Let $\bfg =(g_i)_{i\in I_m}$ be the increment sequence of a uniform polymatroid $\Pol$ 
and let $\GPPD$. If $m\geq 3$ and $g_1>g_{m-1}>0$ and $\min \Delta \neq \set{\set{x}}$ 
for any $x\in J$, then the access structure $\G$ is compartmented.
\end{thm}

\begin{proof}
Let us suppose that there are $x,y\in J$ such that the blocks $P_x$ and $P_y$ are 
hierarchically comparable. According to Proposition \ref{prop:not:comparable}, no 
minimal set in $\Delta$ can contain both $x$ and $y$, so $|X|\leq m-1$ for every 
$X\in \min\Delta$. By assumption $\eta(\bfg)=m$. If $x,y\notin X$ for a certain 
$X\in \min \Delta$, then by Proposition \ref{prop:two:offmin} we obtain 
$g_0=g_1=\ldots=g_{m-1}$, a contradiction.  

If does not exist any set $X$ in $\min \Delta$ such that $x,y\notin X$, then without 
loss of generality we can assume that $x\in X$ and $y\not\in X$ for a certain 
${X\in \min \Delta}$. If $|X|\geq 2$, then by proposition \ref{prop:one:offmin} we get 
$g_1=\ldots=g_{m-1}$, a contradiction again. If $|X|=1$, then 
$\min \Delta = \set{\set{x}}$, as otherwise both $x, y$ would be outside a certain 
minimal set in $\Delta$, but this is not the case now. 
\end{proof}

Let us notice that if  $g_{m-1}>0$, then the above theorem implies that the appearance 
of non-compartmented access structure can be expected in the first row or in the two 
last columns of Table \ref{tab:general}. 
In the next theorem we shall prove that the access structures that appear in the last 
row of Table \ref{tab:general} are compartmented excluding the cells A3, E3 and F3.

\begin{thm}
\label{thm:last:row}
Let $\bfg =(g_i)_{i\in I_m}$ be the increment sequence of a uniform polymatroid $\Pol$ 
and let $\GPPD$. If $m\geq 3$ and $g_1>g_{m-1}$ and $k:=\mu(\Delta)>1$, then the access 
structure $\Gamma$ is compartmented.
\end{thm}

\begin{proof}
Suppose to the contrary that $\Gamma$ is not compartmented, i.e. there are two blocks 
which are hierarchically comparable. For simplicity we assume $P_y \preccurlyeq P_x$ for 
certain $x,y\in J, \ x\neq y$. 
By  Proposition \ref{prop:not:comparable} no set in $\min \Delta$ contains both $x$ and 
$y$, in particular there is a subset of $J$ with $k$ elements which does not belong to 
$\min\Delta$. This and Lemma \ref{lem:head} (1) imply $g_k>0$. Let $n:=\eta(\bfg)$. 
Obviously $2\leq k<n\leq m$ and $g_n=0$. Proposition \ref{prop:not:inferior} implies 
$y\notin X$ for every $X\in \min \Delta$ with $|X|\leq n-1$. 

If there exists $X\in \min \Delta$ such that $|X|=k$ and $x\in X$, then 
$3\leq |X\cup \set{x,y}|=k+1\leq n$. Now we assume that $x\notin X$ for every 
$X\in \min \Delta$ with $|X|=k$. Suppose $g_{k+1}=0$. Let us fix $X\in \min\Delta$ with 
$|X|=k$ and $z\in X$. Let us consider $Z:=(X\setminus\set{z})\cup\set{x,y}$. From Lemma 
\ref{lem:head} (1) we have $Z\in \Delta$ as $|Z|=k+1$. By Proposition 
\ref{prop:not:comparable} the set $Z$ cannot be minimal as it contains $\set{x,y}$, so 
there is $Y\in \min \Delta$ such that $Y\subsetneq Z$ and hence $|Y|=k$. 
Obviously, $x,y\notin Y$ which implies $Y\seq X\setminus \set{z}$, a contradiction.
This shows that $g_{k+1}>0$, so $k+1<n$. Thus we have $4\leq |X\cup\set{x,y}|=k+2\leq n$ 
for arbitrary $X\in \min\Delta$ with $|X|=k$.

In both cases we can apply Corollary \ref{cor:n=m} which implies $g_1=g_{m-1}$, a 
contradiction. This completes the proof.
\end{proof}

The following table presents a general arrangement of multipartite access structure 
determined by monotone increasing families contained in 
$\pow{J}\setminus\set{\emptyset}$ and uniform polymatroids. The cells A1-C1 and A3 do 
not contain any objects since the suitable monotone increasing families and polymatroids 
are not compatible. A monotone increasing family which is not compatible with given 
polymatroid can occur in every cell of the table. A complete overview of hierarchical 
orders of all access structure obtained from uniform polymatroids with $m=4$ can be 
found at Table 2.

\fboxsep=8mm
\begin{table}[ht]
\caption{Hierarchical (pre)orders of access structures obtained form uniform 
polymatroids.}

\centerline{T - threshold, C - compartmented, H - hierarchical, H$^*$ - weakly 
hierarchical}
\label{tab:general}
\begin{center}

\tabcolsep=4pt
\begin{tabular}{|p{5mm}|>{\columncolor{gray!30}}p{2.5cm}|c|c|p{2.6cm}|p{2.6cm}|c|c|}
\hline
		&\cellcolor{white}   &\centering{A} &      \centering{B}    & \centering{C} & 
\centering{D}  & E & F\\ 
\hline
		&  \rule{1.6cm}{0mm} $\bfg$ &   
		\multicolumn{3}{c|}{\cellcolor{gray!30} $g_{m-1}=0$}  & 
		\multicolumn{3}{c|}{\cellcolor{gray!30} $g_{m-1}>0$}\\ 
\hhline{~>{\arrayrulecolor{gray!30}}->{\arrayrulecolor{black}}|------}
		&    &  \multicolumn{2}{c|}{\cellcolor{gray!30} $g_2=0$} &\cellcolor{gray!30} 
\centering{$g_2>0$} &\cellcolor{gray!30} \centering{$g_1>g_{m-1}$} & 
\multicolumn{2}{c|}{\cellcolor{gray!30} $g_1=g_{m-1}$} \\ 
\hhline{~>{\arrayrulecolor{gray!30}}->{\arrayrulecolor{black}}|-|-|>{\arrayrulecolor{gray!30}}-->{\arrayrulecolor{black}}--}
		&  \diag{3}{17mm}		{$\ \ \Delta$}  & \cellcolor{gray!30}   $g_1=0$ 
&\cellcolor{gray!30}  $g_1>0$  & \cellcolor{gray!30} & \cellcolor{gray!30}
&\cellcolor{gray!30} $g_0>g_1$  & \cellcolor{gray!30} $g_0=g_1$    \\
\hhline{>{\arrayrulecolor{black}}--------}\ 1 &  $\begin{array}{c} \mu(\Delta)=1\\ \left|\min\Delta\right|=1\end{array} $	& \multicolumn{3}{c|}{\centering{--}} & \multicolumn{2}{c|}{\centering{H}} &   H$^*$\\ \hhline{|>{\arrayrulecolor{black}}->{\arrayrulecolor{black}}->{\arrayrulecolor{black}}->{\arrayrulecolor{black}}->{\arrayrulecolor{black}}->{\arrayrulecolor{black}} ->{\arrayrulecolor{black}}->{\arrayrulecolor{black}}|-|}\ 2 	& \rule[-0.9cm]{0mm}{2.1cm} $\begin{array}{c} \mu(\Delta)=1\\ 
\left|\min\Delta\right|>1\end{array} $ & T & \centering{C \& H} & \centering{ 
\framebox[26mm]{\textbf{C \& H}}} & \centering{C} & H & H$^*$\\
\hline
\ 3 	& \rule[-1cm]{0mm}{2.1cm} $\ \mu(\Delta)>1\ $& -- & $\begin{array}{c}\mbox{}\\ 
\mbox{C} \\ \mbox{  }\end{array}$ & \centering{C} & \centering{C} & C \& H  & C \& H$^*$ 
 \\
\hline
\end{tabular}
\end{center}
\end{table}

To describe the hierarchical order determined by an access structure $\G$ in a partition 
$\Pi$ of the set of participants $P$ we introduce the following notations $\ord{Y}{X}$ 
and $\pord{Y}{X}$ which are defined as follows. 

\begin{defn}
Let $\Pi =(P_x)_{x\in J}$ be a partition of the set $P$ and let $Y$ and $X$ be disjoint 
subsets of $J$. The hierarchical preorder $\preccurlyeq $ in $\Pi$ is said to be of type 
$\ord{Y}{X}$ if 
$$P_y\preccurlyeq P_x  \ \ \ \ \Longleftrightarrow\ \ \ \ \left(\rule{0pt}{2.5ex} x=y \ \ \mbox{ or } \ \ (y\in Y \mbox{ and } x\in X)\right).$$
In particular, no different blocks are hierarchically equivalent, i.e. $\preccurlyeq$ is 
an order. Moreover, if $X$ or $Y$ is empty, then the order $\ord{Y}{X}$ is compartmented.
\end{defn}

\begin{defn}
Let $\Pi =(P_x)_{x\in J}$ be a partition of the set $P$ and let $X$ and $Y$ be disjoint 
subsets of $J$. The hierarchical preorder $\preccurlyeq $ in $\Pi$ is said to be of type 
$\pord{Y}{X}$ if 
$$P_y\preccurlyeq P_x \ \ \ \ \Longleftrightarrow \ \ \ \ \left(\rule{0pt}{2.5ex} x=y \ \ \mbox{ or } \ \ x,y\in Y \ \ \mbox{ or } \ \ (y\in Y \mbox{ and } x\in X) \right).$$
The preorder $\pord{Y}{X}$ is not an order (unless $|Y|\leq 1$). In particular, 
$P_x, P_y$ are hierarchically equivalent whenever $x,y\in Y$, but no different blocks 
$P_x, P_y$ with  $x,y\in J\setminus Y$ are hierarchically equivalent.
\end{defn}

One can notice that if the set $Y$ is empty, then every two blocks are hierarchically 
independent in the preorder $\pord{\emptyset}{X}$. If the set $X$ is empty, then we get  
$\pord {Y} {\emptyset}$, that means each two blocks are hierarchically equivalent (cf. 
Example \ref{ex:threshold}). Noteworthy is also the observation that if $|Y| \leq 1$, 
then $\ord{Y}{X}= \pord{Y}{X}$. The defined preorders can be presented in the form of 
the following Hasse diagrams.

\usetikzlibrary{shapes}
\begin{center}
\begin{tikzpicture}
[elem/.style={circle,fill=black!50,draw=black,thick,inner sep=2pt,minimum size=5},  
hook/.style={circle,fill=black!0,thin,inner sep=0pt}]

\node[elem,label=above: $P_{x_1}$] (B1)  at (0.1, 1.5) {} ;
\node[elem,label=above: $P_{x_2}$] (B2)  at (0.8, 1.5) {};
\node[elem,label=above: $P_{x_r}$] (B3)  at (1.9, 1.5) {};

\node[elem,label=below: $P_{y_1}$] (A1)  at (0.0, 0.0) {};
\node[elem,label=below: $P_{y_1}$] (A2)  at (0.7, 0.0) {};
\node[elem,label=below: $P_{y_s}$] (A3)  at (2.1, 0.0) {};

\node[elem,label=below: $P_{z_1}$] (C1)  at (2.7, 0.85){};
\node[elem,label=below: $P_{z_2}$] (C2)  at (3.3, 0.85) {};
\node[elem,label=below: $P_{z_t}$] (C3)  at (4.5, 0.85) {};

\node[  hook,label=below: {\large $\ldots$}] (itdD) at (1.4, 0.2)  {};	%lower dots
\node[  hook,label=below: {\large $\ldots$}] (itdG) at (1.4, 1.7)  {};	%upper dots
\node[  hook,label=below: {\large $\ldots$}] (itdS) at (3.9, 1)  {};	%center dots
\node[  hook,label=below: {$\underbrace{\rule{2.7cm}{0mm}}_{(P_y)_{y\in Y}}$}] 
(underbrace) at (1.0,-0.4)  {};  
\node[  hook,label=below: {$\overbrace{\rule{2.4cm}{0mm}}^{(P_x)_{x\in X}}$}] 
(overbrace) at (1.0,2.9)  {};  
%Edges
\draw[->] (node cs:name=A1)-- (node cs:name=B1);
\draw[->] (node cs:name=A1)-- (node cs:name=B2);
\draw[->] (node cs:name=A1)-- (node cs:name=B3);
\draw[->] (node cs:name=A2)-- (node cs:name=B1);
\draw[->] (node cs:name=A2)-- (node cs:name=B2);
\draw[->] (node cs:name=A2)-- (node cs:name=B3);
\draw[->] (node cs:name=A3)-- (node cs:name=B1);
\draw[->] (node cs:name=A3)-- (node cs:name=B2);
\draw[->] (node cs:name=A3)-- (node cs:name=B3);
%Caption
\node[  hook,label=below: {Order of the type $\ord{Y}{X}$}.] (caption1) at (2.1, -1.2)  
{};  
%Second part
\node[elem,label=above: $P_{x_1}$] (B1)  at (7.1, 1.5) {};
\node[elem,label=above: $P_{x_2}$] (B2)  at (7.8, 1.5) {};
\node[elem,label=above: $P_{x_r}$] (B3)  at (8.9, 1.5) {};

\node(group) at (8.0, -0.1) [ellipse,draw] {$P_{y_1}\ldots P_{y_s}$};

\node[elem,label=below: $P_{z_1}$] (C1)  at ( 9.7, 0.85) {};
\node[elem,label=below: $P_{z_2}$] (C2)  at (10.3, 0.85) {};
\node[elem,label=below: $P_{z_t}$] (C3)  at (11.5, 0.85) {};
%dots & braces
\node[  hook,label=below: {\large $\ldots$}] (itdG) at (8.4, 1.7)  {};	%upper dots
\node[  hook,label=below: {\large $\ldots$}] (itdS) at (10.9, 1)  {};	%center dots  
\node[  hook,label=below: {$\underbrace{\rule{2.7cm}{0mm}}_{(P_y)_{y\in Y}}$}] 
(underbrace) at (8.0,-0.4)  {};  
\node[  hook,label=below: {$\overbrace{\rule{2.4cm}{0mm}}^{(P_x)_{x\in X}}$}] 
(overbrace) at (8.0,2.85)  {};  
%Edges
\draw[->] (node cs:name=group)-- (node cs:name=B1);
\draw[->] (node cs:name=group)-- (node cs:name=B2);
\draw[->] (node cs:name=group)-- (node cs:name=B3);
\node[  hook,label=below: {Preorder of the type $\pord{Y}{X}$.}] (caption2) at 
(9.15, -1.2)  {};  
\end{tikzpicture}
\end{center}

Now we want to describe the hierarchical orders of multipartite access structures 
determined by some special type of uniform polymatroids. In Example \ref{ex:threshold} 
we considered the case of polymatroids with the increment sequence $\bfg$ satisfying 
$\eta(\bfg)=1$. We will now deal with  polymatroids $\Pol=(J,h,\bfg)$ with 
$\eta(\bfg)=2$, i.e. $g_0 \geq g_1> g_2 = 0$. The result presented below refers to the 
column B of Table \ref{tab:general}.

\begin{thm}
\label{thm:g2:zero}
Let $\Pol= \Pol(\Pi)=(J,h,\bfg)$ be a uniform polymatroid with 
the increment sequence $\bfg =(g_i)_{i\in I_m}$ such that $\eta(\bfg)=2$.  
\begin{enumerate}
\item If $ g_0> g_1 $, then a monotone increasing family $\Delta \subseteq 
\pow{J}\setminus\set{\emptyset}$ is compatible with the polymatroid $ \Pol $ if and only 
if there is a subset ${X \subseteq J}$ such that 
${\min \Delta = \powk{1}{X} \cup \powk{2}{J\setminus X}}$.
\item If $g_0 = g_1$, then a monotone increasing family 
$\Delta \subseteq \pow{J}\setminus\set{\emptyset}$ is compatible with the polymatroid 
$\Pol$ if and only if there is a subset ${X \subseteq J}$, ${|X|\leq 1}$ such that 
$\min\Delta = \powk{1}{X}\cup \powk{2}{J \setminus X}$.
\item Let $\GPPD$ be the access structure determined by the polymatroid $\Pol$ and the 
monotone increasing family $\Delta \subseteq \pow{J}\setminus\set{\emptyset}$ such that 
$\min\Delta = \powk {1} {X} \cup \powk {2} {J  \setminus X} $ for some $ X \subseteq J$, 
then the hierarchical order induced by $\Gamma$ on the set $\Pi$ is of the type 
$\ord{J\setminus X}{X}$.
\end{enumerate}
\end{thm}

\begin{proof}
(1) and (2). ($\Rightarrow$). If the monotone increasing family $\Delta$ is compatible 
with the polymatroid $\Pol$, then Lemma \ref{lem:head} (1) and the assumptions $g_2=0$ 
shows that all subsets of $J$ with two elements belong to $\Delta$. Thus, all sets in 
$\min\Delta$ have one or two elements. Let $X \subseteq J$ denote the collection of 
those elements that form single-element minimal sets. Hence the remaining sets in $\min 
\Delta$ have two elements and do not contain any elements belonging to $X$. Therefore 
$\min\Delta = \powk{1}{X} \cup \powk{2}{J \setminus X}$. If $g_0=g_1$, then it follows 
from Lemma \ref{lem:one:minimal} that $|X|\leq 1$.

(1) and (2). ($\Leftarrow$). To show the reverse implication, consider the monotone 
increasing family ${\Delta}$ such that $\min\Delta =\powk{1}{X} \cup \powk{2}{J 
\setminus X}$ for some $X \subseteq J$. Let us recall that every set with 2 elements 
belongs to $\Delta$. It is easy to see that $h(Y)=g_0+g_1$ for all $Y\seq J$ with 
$|Y|\geq 2$. We shall apply the Csirmaz Lemma. If $W\seq Y\seq J$ such that 
$W \notin \Delta$ and $Y\in \Delta$, then $W=\emptyset$ and $|Y|\geq 1$ or $|W|=1$ and 
$|Y|\geq 2$. In the former case we have $0=h(W)<g_0\leq h(Y)$ and in the latter case 
$g_0=h(W)<g_0+g_1=h(Y)$. Similarly, if $Y,Z\in \Delta$ and $W=Y\cap Z\notin \Delta$, 
then $|W|\leq1$. If $|W|=1$, then $|Y|,|Z|\geq 2$. Hence 
$$h(W)+h(Y\cup Z)=g_0+(g_0+g_1)<(g_0+g_1)+(g_0+g_1)=h(Y)+h(Z).$$
Now we assume $W=\emptyset$, so $|Y\cup Z|\geq 2$ and $|Y|,|Z|\geq 1$. If $g_0>g_1$, then
$$h(Y\cup Z)=g_0+g_1 < g_0+g_0\leq h(Y)+h(Z).$$
In case $g_0=g_1$ we assumed that there is at most one singleton in $\Delta$, so $|Y|>1$ 
or $|Z|>1$. Hence 
$$h(Y\cup Z)=g_0+g_1 < g_0+g_0+g_1\leq h(Y)+h(Z).$$
Thus both conditions of the Csirmaz Lemma are satisfied, which completes the proof of 
(1) and (2).

(3) If $X=\emptyset$, then $\mu(\Delta)=2$.
For $m=2$ we have $\Delta=\set{J}$ and by Proposition \ref{prop:not:comparable} the 
blocks $P_x$ and $P_y$ are hierarchically independent, so $\G$ is compartmented. 
Assuming $m\geq 3$ we can apply Theorem \ref{thm:last:row} and we conclude that the 
obtained access structure is compartmented, i.e. $\ord{J}{\emptyset}$. 

We now turn to the case $X\neq \emptyset$. 
By assumption $\set{x}\in \min\Delta$ for all $x\in X$, so applying Proposition 
\ref{prop:not:inferior} we see that $P_x$ cannot be hierarchically inferior or 
equivalent to any other $P_z$. In particular, $P_x$ and $P_y$ are mutually 
hierarchically independent whenever $x,y\in X$ and $x\neq y$.
If $x,y\in J\setminus X$, then $\set{x,y}\in \min\Delta$, so $P_x$ and $P_y$ are 
hierarchically independent by Proposition \ref{prop:not:comparable}. In particular if 
$X=J$, then we get the compartmented order $\ord{\emptyset}{J}$. 

It remains to show that $P_y \prec P_x$ for $x \in X$ and $y \in J \setminus X$ with 
$\emptyset\subsetneq X\subsetneq J$. Let $\W$ be a minimal vector in $\G$ such that 
$w_y \neq 0$. If such a vector does not exist, then the block $P_y$ is redundant 
so $P_y \prec P_x$. Otherwise, applying Lemma \ref{lem:supp:property} (3) yields 
$\W\in\B(\Pol,\supp \W)$ and $\supp \W \in \Delta $. Note that 
$\set{y}\notin \min \Delta$, so $|\supp \W| \geq 2 $. According to Equation 
\eqref{eq:equal_h} we get $\W \in \B (\Pol, \supp \W) \subseteq \B (\Pol, J )$. 
Lemma \ref{lem:onestep} shows that $\W':=\W-\E{y}+\E{x} \in \B(\Pol,J)$ or there exist
a set $Y \subseteq J \setminus \set{y}$, $x \in Y$ and a vector $\V \in \B (\Pol,Y)$, 
such that $\V \leq \W'$. In the first case we have $ \W '\in \G $ by Lemma 
\ref{lem:supp:property} (1). 
If the second case occurs, then we note that $\set{x}\seq Y \in \Delta$, so from Lemma 
\ref{lem:supp:property} (1) we get $ \V \in \G $, hence $ \W '\in \G $. This proves that 
$P_y \prec P_x$. In this way we showed that the order on the set $(\Pi, \preccurlyeq)$ 
is of the type $\ord{J\setminus X}{X}$.
\end{proof}

\begin{rem}
The above theorem combined with Example \ref{ex:threshold} is strong enough to classify 
all bipartite access structure determined by uniform polymatroids with $m=2$. If 
$\eta(\bfg)=1$, then we have a threshold access structure (cf. Example 
\ref{ex:threshold}). If $\eta(\bfg)=2$, then we consider three monotone increasing 
families $\Delta_1=\set{\set{x},J}$, $\Delta_2=\set{\set{x},\set{y},J}$ and and 
$\Delta_3=\set{J}$ of subsets of $J:=\set{x,y}$. Let us note that 
$\Delta'_1=\set{\set{y},J}$ is hierarchically equivalent to $\Delta_1$ as it can be 
obtained by the permutation of $x$ and $y$. It is easy to see that $\Delta_1$ is 
compatible with every polymatroid $(J,h,\bfg)$ with $g_1>0$ and the resulting access 
structures induce on $\set{P_x,P_y}$ the order of the type $\ord{\set{y}}{\set{x}}$. 
Moreover, $\Delta_2$ is not compatible with a polymatroid such that $g_0=g_1$ but in the 
remaining cases the resulting access structures are compartmented.  
\end{rem}

The following theorem describes the hierarchy on the access structures determined by 
uniform  polymatroids with $g_0=g_1 = \dots = g_{m-1}>0$. This result corresponds to the 
access structures  that appear in the column F of Table \ref{tab:general}.

\begin{thm}
\label{thm:last:column}
Let $\Pol=(J,h,\bfg)$ be a uniform polymatroid with $m:=|J|\geq 3$ and the increment 
sequence $\bfg =(g_i)_{i\in I_m}$ such that $g_0 = g_{m-1}>0$. 

\begin{enumerate}[(1)]
\item A monotone increasing family $\Delta\subseteq \pow{J}\setminus\set{\emptyset}$ 
is compatible with the polymatroid  $\Pol$  if and only if $\min \Delta = \set{X} $ 
for a certain $\emptyset\neq X\subseteq J$. 
\item Let $\GPPD$ be the access structure determined by the polymatroid $\Pol$ and the 
monotone increasing family $\Delta \subseteq \pow{J}\setminus\set{\emptyset}$ such 
that $\min \Delta = \set{X} $ for a certain $\emptyset\neq X\subseteq J$, 
then
\begin{enumerate}[(a)]
\item[(2a)] The vector $\sum_{x \in X } g_0\E{x}$ is the only minimal authorized 
vector in the access structure $\G$.
\item[(2b)] The hierarchical preorder induced by $\Gamma$ on the set $\Pi$ is of the 
type $\pord{J\setminus X}{X}$.
\end{enumerate}
\end{enumerate}
\end{thm}

\begin{proof}
(1) Since $g_0=g_{m-1}$, i.e. $\eta(\bfg)=m$ we can apply Lemma \ref{lem:one:minimal}  
which implies that if $\Delta$ is compatible with the polymatroid $\Pol$, then 
$\min\Delta$ contains only one set. To prove the reverse implication it is enough to 
apply Lemma \ref{lem:compatible}. 

(2a) We apply Lemma \ref{lem:min:coeffic} (1) for an arbitrary $Y\in \Delta$ and an 
arbitrary $\W\in \B(\Pol,Y)$. For $l:=|Y|$ we have $h_1\geq w_z\geq g_{l-1}=g_0=h_1$  
and consequently $w_z = h_1=g_0$ for every $z \in Y$. Since $X\subseteq Y$, so  
$\W\geq \sum_{x \in X } g_0\E{x}$ for every set $Y\in \Delta$ and for every vector 
$\W\in \B(\Pol, Y)$. This shows that the vector $\sum_{x \in X } g_0\E{x}$ is the only 
minimal authorized vector.

(2b) According to Proposition \ref{prop:not:comparable} the blocks indexed by the 
elements in $X$ are hierarchically independent. In particular, if $X=J$, then 
the hierarchical order on $\Pi$ induced by $\Gamma$ is of the type $\pord{\emptyset}{J}$.

Now we assume $|X|<m$. It is shown above that $\sum_{x \in X } g_0\E{x}$ is the only 
minimal authorized vector, so all the blocks $P_y$ with $y\notin X$ are redundant. In 
particular, they are mutually hierarchically equivalent and every block $P_x, \ x\in X$ 
is hierarchically superior but not equivalent to $P_y, \ y\in J\setminus X$, which 
follows from Proposition \ref{prop:not:inferior}. Moreover, all blocks in 
$\set{P_x : \ x\in X}$ are hierarchically independent by Proposition 
\ref{prop:not:comparable}. We conclude, that the hierarchical order of $\Pi$ induced 
by $\Gamma$ is of the type $\pord{J\setminus X}{X}$. 
\end{proof}

Now we shall prove a similar theorem which describes hierarchical order of access 
structures determined by uniform polymatroids with ${g_0>g_1 = \dots = g_{m-1}>0}$ and 
monotone increasing families compatible with them. This theorem describes access 
structures located in the column E of Table \ref{tab:general}.

\begin{thm}
\label{thm:last-but-one:column}
Let $\Pol=(J,h,\bfg)$ be a uniform polymatroid with the increment sequence 
$\bfg =(g_i)_{i\in I_m}$ such that $m:=|J|\geq 3$ and ${g_0>g_1 = g_{m-1}>0}$.

\begin{enumerate}[(1)]
	\item  A monotone increasing family $\Delta\subseteq \pow{J}\setminus\set{\emptyset}$ 
	is compatible with $\Pol$ if and only if $\min \Delta = \set{X} $ for a certain 
	$X\subseteq J$ or $\min \Delta = \powk{1}{J}$. 

	\item  Let $\GPPD$ be the access structure determined by the polymatroid $\Pol$ and 
	the monotone increasing family $\Delta \subseteq \pow{J}\setminus\set{\emptyset}$. 
	Then 
\begin{enumerate}[(a)]
	\item[(2a)] If $\min \Delta = \set{X}$ for a certain $\emptyset\neq X\subseteq J$, 
	then the hierarchical order induced by $\Gamma$ on $\Pi$ is of the type  
	$\ord{J\setminus X}{X}$.
	\item[(2b)] If $\min \Delta = \powk{1}{J}$, then the hierarchical order induced by 
	$\Gamma$ on $\Pi$ is of the type $\ord{\emptyset}{J}$.
\end{enumerate}
\end{enumerate}
\end{thm}

\begin{proof}
(1) Let us assume that $\Delta$ is compatible with $\Pol$. It is enough to consider the 
case where $\Delta$ has at least two different minimal sets. From the assumption 
$g_{m-1}>0$ we have $\eta(\bfg)=m$, so applying Lemma \ref{lem:one:minimal} we conclude 
that those sets must be singletons. Let $\set{x}, \set{y}\in\min\Delta$ for some 
$x, y \in J$. Suppose that there is $z \in J$ such that $\set{z}\not\in\min\Delta$. Of 
course $\set{x,z},\set{y,z}\in\Delta$, but 
$\set{x,z}\cap(\set{y,z} =\set{z}\not\in\Delta$. 
Using the Csirmaz Lemma yields
\[ h(\set{z}) + h(\set{x,y,z}) <h(\set{x,z }) + h(\set{y, z }), \]
hence we get $h_3 - h_2 <h_2-h_1$ and consequently $g_2 <g_1$ which is a contradiction, 
so every singleton belongs to $\min \Delta$. To show the reverse implication, let us 
consider two cases:

If $\min \Delta = \set{X}$ for some $X \subseteq J$, then we refer to Lemma 
\ref{lem:compatible}.

If $\min \Delta = \powk{1}{J}$, then the claim it follows from Lemma 
\ref{lem:k-compatible}.

(2a) Assume $ \min \Delta = \set{X}$ for some $\emptyset\neq X \subseteq J$. The fact 
that $P_x, P_y$  are hierarchically independent for arbitrary 
${x, y \in X, \ \ x\neq y}$ is obtained directly from Proposition 
\ref{prop:not:comparable}. In particular, if $ X = J $, then the ordered set $(\Pi, 
\preccurlyeq)$ is of the type $\ord {\emptyset}{J}$.

Now we assume that $|X|<m$. Consider $x\in X$ and $y \not \in X$. According to 
Proposition \ref{prop:not:inferior} the blocks $ P_y $ and $ P_x $ are hierarchically 
independent or $P_y \prec P_x$. We shall show that $P_y \preccurlyeq P_x$.  

Let us assume that $\W$ is an arbitrary minimal vector in $\G$ such that $w_y \neq 0$. 
The existence of such vectors is ensured by Theorem \ref{thm:connected}.
Then from Lemma \ref{lem:supp:property} (3) we have $\W \in \B (\Pol, \supp \W)$ and 
${\supp \W \in \Delta}$, so $X \subseteq \supp \W $, in particular $ {x\in \supp \W}$. 
Note that $k:=|\supp\W|\geq 2$, because $ y, x \in \supp \W $. According to Lemma 
\ref{lem:min:coeffic} (1), we get $w_y \geq g_ {k-1}$. By assumption we have 
$g_{k-1}=g_1$, hence we can consider two cases:

(i) $ w_y = g_1 $, so according to Lemma \ref{lem:min:coeffic} (2) we have 
$\V: = \W-w_y \E{y} \in \B (\Pol, \supp \W \setminus \set{y })$, but 
$X \subseteq \supp \W \setminus \set{y}$, so from Lemma \ref{lem:supp:property} (1) we 
get $\V\in\G$. Then of course $\V\leq \W':= \W-\E{y}+\E{x}$, so $\W' \in \G$.

(ii) $ w_y> g_1 $ and denote $ Y: = \supp \W $. According to Lemma \ref{lem:onestep} we 
get $\W': = \W - \E{y} + \E{x} \in \B (\Pol, Y)$ or there is a set 
$Z \subseteq Y \setminus \set{y}, \ x \in Z$ such that $\V: = \W_Z \in \B (\Pol, Z)$. In 
particular, we have $ |\W_Z| = |\V| = h_{|Z|}$. Let $Y=Z\cup W\cup\set{y}$ be the union 
of three disjoint sets, where $W=Y\setminus(Z\cup\set{y})$. Then, using Lemma 
\ref{lem:min:coeffic} (1) and assumptions, we obtain that each coordinate of the vector 
$\W$ is at least $ g_1 $, hence:
$$
h_{|Y|} = |\W| = |\W_Z| + |\W_W| + w_y > |\V| + |W| g_1 + g_1 = h_{|Z|} + |W|g_1 + g_1 = 
h_{|Y|},
$$
where the last equality is obtained from Equation \eqref{eq:sums} in the following way:
$$h_{|Y|} - h_{|Z|} = \sum_{i=|Z|}^{|Y|-1} g_i = |W|g_1 + g_1.$$
A contradiction we have obtained shows that $\W'\in \B(\Pol, Y)\seq \G$. In both of the 
above cases we have received that $\W'\in \G$. Since this holds for all $\W\in 
\min\Gamma$ with $w_y>0$, we conclude $P_y \preccurlyeq P_x$.

It remains to show that $P_y,P_x$ are hierarchically independent when $x,y \not \in X $. 
If it were otherwise, then assuming $n:=m$ and applying Proposition 
\ref{prop:two:offmin}, we would get $g_0=g_1 $ contrary to the assumption made here. 
In this way, we showed that the order on $\Pi$ is of type ${\ord {J \setminus X} {X}}$.

(2b) Assume $\min \Delta =\powk{1}{J}$. If $P_y \preccurlyeq P_x$ for some 
${x, y \in J}$, then $P_y$ is hierarchically inferior or equivalent to $P_x$ and 
$\set{x},\set{y}\in \min \Delta$, which contradicts Proposition \ref{prop:not:inferior}. 
In this way we showed that the order on $\Pi$ is of type $\ord{\emptyset}{J}$.
\end{proof}

\begin{rem}
\label{rem:k:regular}
The result of (2b) can be generalized to all monotone increasing families $\Delta$ with 
$\min\Delta=\powk{k}{J}$, where $k=1,\ldots,m$. According to Lemma 
\ref{lem:k-compatible} (cf \cite[Lemma 6.1]{FPXY}) such $\Delta$ is compatible with a 
uniform polymatroid $\Pol=(J,h,\bfg)$ if and only if $g_{k-1}>g_k$. Let $\GPPD$. If 
$k=1$ and $g_1>0$, then $\Gamma$ is compartmented by Proposition 
\ref{prop:not:inferior}. If $k\geq 2$, then $\mu(\Delta)\geq 2$, so $\Gamma$ is 
compartmented which follows from Theorem \ref{thm:last:row}. In both cases the 
hierarchical order induced by $\Gamma$ in $\Pi$ is of the type $\ord{\emptyset}{J}$. 
A similar class, called uniform multipartite access structures were also considered by 
Farr\`as, Padr\'o, Xing and Yang \cite{FPXY}. 
\end{rem}

Another theorem describes the hierarchy of blocks in the access structures determined by 
polymatroids, for which $g_{m-1}>0$ and monotone increasing families with one minimal 
set which contains exactly one element. This theorem deals with the existence and 
hierarchy of access structures placed in the first row of Table \ref{tab:general}.

\begin{thm} 
\label{thm:singleton} 
Let $\Pol=(J, h, \bfg)$ be a  uniform polymatroid with the increment sequence 
$\bfg =(g_i)_{i\in I_m}$  and let $\Delta\subseteq \pow{J}\setminus\set{\emptyset}$ be a 
monotone increasing family such that $\min \Delta=\set{\set{x}}$ for a certain $x\in J$. 
\begin{enumerate}[(1)]
\item Then $\Delta$ is compatible with the polymatroid $\Pol$ if and only if $g_{m-1}>0$.
\item Let $\GPPD$ be the access structure determined by the polymatroid $\Pol$ such that 
$g_{m-1}>0$ and the monotone increasing family $\Delta$. 
Then
\begin{enumerate}[(a)]
\item[(2a)] If $g_0=g_{m-1}$, then the hierarchical order induced by $\Gamma$ on $\Pi$ 
is of the type $\pord{J\setminus \set{x}}{\set{x}}$. 
\item[(2b)] If $g_0>g_{m-1}$, then the hierarchical order induced by $\Gamma$ on $\Pi$ 
is of the type $\ord{J\setminus \set{x}}{\set{x}}$. 
\end{enumerate}
\end{enumerate}
\end{thm}

\begin{proof}  
(1) The fact that $\Delta$ is compatible with $\Pol$ can be obtained from Lemma 
\ref{lem:compatible}. Conversely, let us suppose that $g_{m-1}=0$. Then every subset of 
$J$ with $m-1$ elements belongs to $\Delta$ by Lemma \ref{lem:head} (1). But this 
contradicts the fact that $J\setminus \set{x}\not \in \Delta$. This implies $\Delta$ is 
not compatible with $\Pol$ whenever $g_{m-1}=0$.

(2) We shall show that $P_y\prec P_x$ for every $y\in J\setminus\set{x}$. 
From Proposition \ref{prop:not:inferior} it follows that $P_x$ is not hierarchically 
inferior to any block in $\Pi$ so $P_x$ is not hierarchically equivalent to any other 
block. Let us fix $y \in J, \ y \neq x$. If $w_y=0$ for every minimal vector $\W\in \G$, 
then the block $P_y$ is redundant, so $P_y \preccurlyeq P_x$. Let us assume that $\W$ is 
a minimal vector in $\G$ such that $w_y \neq 0$. Then from Lemma 
\ref{lem:supp:property} (3) we have $\W \in \B (\Pol, \supp \W)$ and 
${\supp \W \in \Delta}$, so ${x\in \supp \W}$. From Lemma \ref{lem:onestep} it follows 
that ${\W': = \W- \E{y} + \E{x} \in \B (\Pol, \supp \W)}$ or there is a set $Y \subseteq 
\supp \W \setminus \set{y}$, $x \in Y$ such that $\V:=\W_Y \in \B (\Pol, Y)$. In the 
former case, we get $ \W '\in \G $ from Lemma \ref{lem:supp:property} (1). If the latter 
case is fulfilled, then we notice that $Y \in \Delta$, so from Lemma 
\ref{lem:supp:property} (1) we get $\V\in\G$ and $\V\leq \W'$. This means that in both 
cases $\W '\in \G $. Since this holds for all $\W\in \min\Gamma$ with $w_y>0$, we 
conclude $P_y \preccurlyeq P_x$. This shows that the (pre)order on $\Pi$ is of the 
type $\pord{J \setminus \set{x}}{\set{x}}$ or $\ord{J\setminus\set{x}}{\set{x}}$.

(2a) Since $g_0=g_{m-1}$, applying Theorem \ref{thm:last:column} (2b) yields the claim.

(2b) If the preorder on $\Pi$ were of the type $\pord{J\setminus\set{x}}{\set{x}}$, then 
the blocks $ P_y $ and $ P_z $ would be hierarchically comparable  for some $ y, z \in J 
\setminus \set{x}$. This and Proposition \ref{prop:two:offmin}, for $n=m$ imply $ g_0 = 
g_1 = \dots = g_{m-1}$. But this is a contradiction to $g_0>g_{m-1}$, so the order on 
$\Pi$ is of type $\ord{J \setminus \set{x}}{\set{x}}$.
\end{proof}

Let us noticing that the order induced by the access structure satisfying the 
assumptions of the above theorem corresponds to the organizational chart of an 
institution composed of several mutually independent departments 
$(P_y)_{y\in J\setminus \set{x}}$ managed by one superior unit $P_x$. It follows 
from Theorem \ref{thm:singl:delta} that such access structures is ideal.

The results in this chapter provide information about hierarchy induced on the set of ed 
on the set of participant by various access structures contained in Table 1, except the 
cell C2. That area contains objects obtained from monotone increasing families 
$\Delta\seq\pow{J}\setminus\set{\emptyset}$ with $\mu(\Delta)=1$ and  
compatible polymatroids $\Pol=(J,h,\bfg)$ with $3\leq\eta(g)\leq m-1$. Computer 
calculations show that this area contains both compartmented and hierarchical access 
structures and some of them are different of those considered in the above theorems. 
Some examples can be seen in Table 2.

Every linearly ordered subset of a partially (pre)ordered set is called a \emph{chain}. 
A chain that contains only one element is referred to as \emph{trivial}.
We assume that a chain in a partition of participants does not contain hierarchically 
equivalent blocks. Let us observe that every non-trivial chain of blocks in the access 
structures investigated above contains 2 blocks. The next theorem shows that all weakly 
hierarchical access structures obtained from uniform polymatroids have this property.

\begin{thm}
\label{thm:height}
Every chain in the hierarchical access structure determined by arbitrary uniform 
polymatroid contains 1 or 2 blocks.
\end{thm}

\begin{proof}
Let $n:=\eta(\bfg)$. For $n=1$, i.e. $g_0>g_1=0$, it follows from Example 
\ref{ex:threshold} that  $\GPPD$ is a threshold access structure, so all blocks of 
participants are mutually hierarchically equivalent, thus every chain is trivial.

Suppose that $n\geq 2$ and $\Pi$ contains a chain of blocks composed of three 
hierarchically non-equivalent blocks, i.e. $P_z \prec P_y \prec P_x$ for some 
$x, y, z \in J$. Let $X\seq J$ such that $|X|=n$ and $y,z \in X$. By Lemma 
\ref{lem:head} (1) we have $X\in \Delta$, but Proposition \ref{prop:not:comparable} 
implies that $X\notin \min \Delta$. Thus there is $Y\subsetneq X$ such that $Y\in \min 
\Delta$ in particular $|Y|<n$. If $n=2$, then $|Y|=1$, but neither $\set{y}$ nor 
$\set{z}$ is minimal in $\Delta$, which follows from Proposition 
\ref{prop:not:inferior}, a contradiction. If $n\geq 3$, then by Proposition 
\ref{prop:not:inferior} we know that $y,z\notin Y$. Thus $|Y|\leq n-2$. 
Using Proposition \ref{prop:two:offmin} we get $g_0=g_{n-1}$ and this combined with 
Corollary \ref{cor:n=m} shows that $g_0=g_{m-1}$. Now from Theorem
\ref{thm:last:column} we conclude that every chain in $\Pi$ contains at most 2 blocks, 
which contradicts our assumption.
\end{proof}

The above theorem seems quite surprising, because for other polymatroids one can 
construct hierarchical access structures with chains of arbitrary length. For instance, 
such objects can be found in \cite{FP}, \cite{FPXY}, \cite{tT}, \cite{mK} and others.

%================================================================
\section{Ideal access structures obtained from uniform polymatroids}
\label{sec:ideal:acc:str}
%================================================================

In this section we shall prove that the access structure studied in theorems \ref{thm:g2:zero}, \ref{thm:last:column}, \ref{thm:last-but-one:column} and \ref{thm:singleton} are ideal. To do this we show that all simple  extensions of suitable uniform polymatroids are representable over sufficiently large finite fields and then we apply Remark \ref{rem:represent}.
We begin by recalling Example \ref{ex:threshold} where we noticed that every polymatroid 
$\Pol=(J,h,\bfg)$ with $\eta(\bfg)=1$ determines a threshold access structure which is 
known to be ideal as it is realized by the Shamir threshold secret sharing scheme. Now 
we shall consider the case $\eta(\bfg)=2$.

\begin{thm}
All access structures determined by any uniform polymatroid $\Pol=(J,h,\bfg)$ with 
$\eta(\bfg)=2$ are ideal.
\end{thm}

\begin{proof}
The assumption $\eta(\bfg)=2$ implies $g_0\geq g_1>g_2=0$. 
Let $\Delta\seq\pow{J}\setminus\set{\emptyset}$ be a monotone increasing family 
compatible with $\Pol$. It is enough to show that the simple extension $\Pol'$ of $\Pol$ 
induced by $\Delta$ is a representable polymatroid. Let $\KK$ be a finite field with 
$q:=|\KK|>m$. By an abuse of notation, we will use $\theta$ to denote the zero vector in any vector space $\KK^n$.
Let us consider a collection $(a_x)_{x\in J}$ of pairwise different nonzero elements of $\KK$. For every $x\in J$ we define $V_x:=\set{(\alpha,a_x\alpha)\ :\ \alpha\in \KK^{g_1}}$. It easy to check that $V_x$ is a vector subspace of $\KK^{g_1}\times\KK^{g_1}$ and $\dim V_x=g_1$. Assume $x\neq y$ and $(\alpha_1,\alpha_2)\in V_x\cap V_y$. Hence $\alpha_2=a_x \alpha_1$ and $\alpha_2=a_y \alpha_1$, so  $\theta=a_x \alpha_1-a_y \alpha_1=(a_x-a_y)\alpha_1$. Since $a_x-a_y\neq 0$, so $\alpha_1=\theta$. This shows $V_x\cap V_y=\set{\theta}$. Hence $\dim(V_x +V_y)=
\dim V_x+\dim V_y-\dim(V_x\cap V_y)=\dim V_x+\dim V_y=2g_1$. In particular, $V_x+V_y=
\KK^{g_1}\times\KK^{g_1}$ for all $x,y\in J,\ x\neq y$. Thus $(V_x)_{x\in J}$ is a vector space representation of the polymatroid $\Pol$ provided $g_0=g_1$. According to Theorem \ref{thm:g2:zero} (2) we have two cases.
If $\min\Delta=\set{\set{x}}\cup \powk{2}{J\setminus\set{x}}$, then we take $\theta\neq \beta\in V_x$. For a certain $x_0\notin J$ we define $V_{x_0}:=\Span(\beta)$. It is easily seen that $(V_x)_{x\in J\cup\set{x_0}}$ is a vector space representation 
of $\Pol'$ induced by $\Delta$. 

If $\min\Delta=\powk{2}{J}$, then we take $\beta\in \KK^{g_1}\times\KK^{g_1}\setminus \bigcup_{x\in J} V_x$. It is possible as $|\bigcup_{x\in J} V_x|\leq mq^{g_1}<q^{g_1+1}\leq q^{2g_1}=|\KK^{g_1}\times\KK^{g_1}|$.  Now we define 
$V_{x_0}:=\Span(\beta)$. It is easily seen that $(V_x)_{x\in J\cup\set{x_0}}$ is a 
vector space representation of $\Pol'$ induced by $\Delta$. 

Now we assume $g_0>g_1$ and define $U_x:=\KK^{g_0-g_1}\times V_x\subseteq 
\KK^{g_0-g_1}\times\KK^{g_1}\times\KK^{g_1}$ for every $x\in J$. For simplicity of notation, the vector space $\KK^{g_0-g_1}\times\KK^{g_1}\times\KK^{g_1}$ will be identified with $\KK^{g_0+g_1}$. 
It is clear that $\dim U_x=g_0$. Moreover $U_x+U_y= (\KK^{g_0-g_1}\times V_x)+(\KK^{g_0-g_1}\times V_y)=\KK^{g_0+g_1}$ and $U_x\cap U_y=\KK^{g_0-g_1}\times \{\theta\}\times \{\theta\}$. In particular $\varepsilon:=(1,0,\ldots,0)\in U_x$ for all $x\in J$. 

If $\Delta$ is compatible with $\Pol$, then by Theorem \ref{thm:g2:zero}.1 there is 
$X\subseteq J$ such that $\min \Delta=\powk{1}{X}\cup\powk{2}{J\setminus X}$.

To explain the general idea of the next step of the proof we use projective geometry. Every subspace $U_x$ can be considered as $(g_0-1)$-dimensional subspace of the projective space of dimension $g_0+g_1-1$. The projective point $E:=\Span(\varepsilon)$ belongs to the intersection of all subspaces $U_x$ (Figure 1). Now we take a projective point $B:=\Span(\beta^*)$ that does not belong to any subspace $U_x$ and the translation of the whole space $\varphi$ sending $E$ to $B$. Then the family of $(\varphi(U_x))_{x\in X}$ together with the family $(U_x)_{x\in J\setminus X}$ form another vector space representation of $\Pol$ (Figure 2). Now we only need to add  $U_{x_0}:=\Span(\beta^*)$ to those families to get a representation of $\Pol'$.

\begin{center}
\begin{tikzpicture}
%left part
\draw [-] (0,1.2) --(2.4,1.2);
\node at (2.6,1.2) {\footnotesize $U_1$} ;
\draw [-] (0.3,0.3) -- (2.1,2.1);
\node at (2.3,2.3) {\footnotesize $U_2$};
\draw [-] (1.2,0) -- (1.2,2.4);
\node at (1.2,2.6) {\footnotesize $U_3$};
\draw [-] (2.1,0.3) -- (0.3,2.1);
\node at (0.1,2.3) {\footnotesize $U_4$};
\node[circle,fill=black,thin,inner sep=1pt] at (1.2,1.2) {};
\node at (1.35,1.4) {\footnotesize $E$};
%right part - up
\begin{scope}[xshift=5.5cm,yshift=0.2cm]
\draw [-] (0,1.2) --(2.4,1.2);
\node at (2.6,1.2) {\footnotesize $U_1$} ;
\draw [-] (1.2,0) -- (1.2,2.4);
\node at (1.2,2.6) {\footnotesize $U_3$};
\node[circle,fill=black,thin,inner sep=1pt] at (1.2,1.2) {};
\node at (1.35,1.4) {\footnotesize $E$};
\end{scope}

%right part - down
\begin{scope}[xshift=5.65cm,yshift=-0.1cm]
\draw [-] (0.3,0.3) -- (2.1,2.1);
\node at (2.3,2.3) {\footnotesize $\varphi(U_2)$};
\draw [-] (2.1,0.3) -- (0.3,2.1);
\node at (0.1,2.3) {\footnotesize $\varphi(U_4)$};
\node[circle,fill=black,thin,inner sep=1pt] at (1.2,1.2) {};
\node at (1.4,1.2) {\footnotesize $B$};
\end{scope}

\node at (1.2,-0.4) {Figure 1};
\node at (6.8,-0.4) {Figure 2};
\end{tikzpicture}
\end{center}

Now we can do the formal calculations. Let $\nu : \KK^{g_0+g_1} \longrightarrow \KK$ be 
defined by $\nu(\alpha)= \nu(a_1,\ldots,a_{g_0+g_1})=a_1$ for every 
$\alpha=(a_1,\ldots,a_{g_0+g_1})\in\KK^{g_0+g_1}$.
Let $\beta_1\in \KK^{g_1}\times \KK^{g_1}\setminus \bigcup_{x\in J} V_x$ and 
$\beta:= (\theta , \beta_1)\in \KK^{g_0+g_1}$. Obviously $\beta\notin U_x$ for 
every $x\in J$.

Now we define $\varphi : \KK^{g_0+g_1} \longrightarrow \KK^{g_0+g_1}$ by setting 
$\varphi(\alpha)=\alpha+\nu(\alpha)\beta$ for all $\alpha\in \KK^{g_0+g_1}$.
Let us notice that $\varphi$ is an isomorphisms of vector spaces, so 
$\dim\varphi(U_x)=\dim U_x=g_0$. Moreover $\varphi(\alpha)=\alpha$ for all $\alpha\in \set{\theta}\times \KK^{g_1}\times \KK^{g_1}$ and 
$\beta^*:=\varphi(\varepsilon)=\varepsilon+\beta\notin U_x$ for all $x\in J$. 

Let $x_0$ be any element not in $J$ and let $U_{x_0}:=\Span(\beta^*)$.
 Then the family $(\varphi(U_x))_{x\in X}\cup (U_x)_{x\in J\setminus X}\cup \set{U_{x_0}}$ is a vector space representation of the simple extension of $\Pol$ 
 induced by $\Delta$. Indeed, if $x\notin X$, then $h(\set{x,x_0})=\dim (U_x+U_{x_0})>\dim U_x=h(\set{x})=g_0$ as $\beta^*\notin U_x$. Thus $\set{x}\notin \min\Delta$. For $x\in X$ we have $h(\set{x,x_0})=\dim (\varphi(U_x)+U_{x_0})=\dim U_x=h(\set{x})=g_0$, so $\set{x}\in \min\Delta$. 

From the fact that $\varphi$ is a vector space isomorphism it follows 
$\varphi(U_x)+\varphi(U_y)=\varphi(U_x+U_y)=\KK^{g_0+g_1}$ for all $x,y\in X$. 
For $x\in X$ and $y\in J\setminus X$ we have $\varphi(U_x)+U_y\supseteq \varphi(\set{\theta}\times V_x)+U_y= (\set{\theta}\times V_x)+U_y= 
\KK^{g_0-g_1}\times (V_x+V_y)=\KK^{g_0+g_1}$. 
In every case $h(\set{x,y,x_0})=\dim(\KK^{g_0+g_1}+U_{x_0})=g_0+g_1$, i.e., 
$\set{x,y}\in \Delta$. If $x,y\notin X$, then $\set{x,y}\in\min\Delta$.
\end{proof}

In the next proof we will need the following well-known property of vector spaces over 
finite fields. Let $V_1,\ldots, V_n$ be proper subspaces of a vector space $V$ over a 
finite field $\KK$. If $|\KK|>n$, then $V_1\cup \ldots \cup V_n\neq V$. Let us recall 
that every uniform polymatroid is representable.

\begin{thm}
\label{thm:singl:delta}
All access structures determined by any uniform polymatroid $\Pol=(J,h,\bfg)$ with 
$\eta(\bfg)=m$ and monotone increasing family $\Delta\subsetneq \pow{J}$ such that 
$|\min \Delta|=1$ are ideal.
\end{thm}

\begin{proof}
Let $\min \Delta=\set{X}$ for a suitable $\emptyset\neq X\subseteq J$ and let $k:=|X|$.
The assumption $\eta(\bfg)=m$ is equivalent to $g_{m-1}>0$ and this implies $h(Y)<h(Z)$ for all $Y\subsetneq Z\subseteq J$. It follows form Lemma \ref{lem:compatible} that $\Delta$ is  compatible with $\Pol$. 

Let $\KK$ be a finite field and let $(V_x)_{x\in J}$ be a $\KK$-vector space 
representation of $\Pol=(J,h,\bfg)$. Then $V_x$ are subspaces of the vector space 
$\KK^{h_m}$ and $\dim V_x=h_1=g_0$ for every $x\in J$. Given any $Y\subseteq J$ we 
define $V_Y:=\sum_{y\in Y} V_y$. If $Y\in \Delta$, then $X\subseteq Y$ and 
$V_X\subseteq V_Y$. If $Y\notin\Delta$, then $X\not\subset Y$ and so $|X\cup Y|>|Y|$.
Hence $\dim(V_X+V_Y)=\dim V_{X\cup Y}=h(X\cup Y)>h(Y)=\dim V_Y$. This shows, that 
$V_X\not\subset V_Y$. Thus $Y\in \Delta$ if and only if $V_X\subseteq V_Y$. 
Since $V_Y\cap V_X$ is a proper subspace of $V_X$ whenever $Y\notin\Delta$ and, so 
assuming $|\KK|>2^m-2^{m-k}$ we have $V_X\cap \bigcup_{Y\in \pow{J}\setminus\Delta}\ V_Y=\bigcup_{Y\in \pow{J}\setminus\Delta}(V_X\cap V_Y)$ is a proper subset of $V_X$. 
This shows that there is $\beta\in V_X$ such that $\beta \notin V_Y$ for all $Y\notin \Delta$. Setting $V_{x_0}:=\Span(\beta)$ we get $(V_{x})_{x\in J\cup\set{x_0}}$ which is 
a vector space representation of the simple extension of $\Pol$ induced by $\Delta$.
\end{proof}

The above proof is not constructive. Using the fact that every uniform polymatroid 
is a sum of uniform matroids one can efficiently build a vector space representation 
of $\Pol$ and then determine a vector $\beta$ that spans the space $V_{x_0}$ but the 
calculations are more complicated. A general outline of this procedure is sketched out 
in \cite[sections III and VI]{FPXY}. 

Let us notice that if $\set{X}=\min \Delta$ then $X$ determines a set of distinguished 
blocks, whose representatives must be present in all authorized sets. Indeed, if 
$\V\in \Gamma$ is an authorized vector, then $\supp \V\in \Delta$, so 
$X\subseteq \supp\V$, thus $v_x\neq 0$ for all $x\in X$. If $|X|\geq 2$ then the access 
structures $\Gamma$ is compartmented by Theorem \ref{thm:last:row}, so all blocks are 
mutually hierarchically independent.

For the sake of completeness, we recall the following result obtained by Farr\`as, 
Padr\'o, Xing and Yang in \cite{FPXY} who characterized the uniform multipartite access 
structures mentioned in Remark \ref{rem:k:regular} and proved that they are ideal. 
Contrary to the above case all participants in any uniform access structure have the 
same rights but different blocks are hierarchically independent. 
Here we reformulate that result as follows. 

\begin{thm}{\cite[Lemma 6.2]{FPXY}}
\label{thm:fpxy}
If the monotony increasing family $\Delta\subseteq \pow{J}$ such that 
$\min\Delta=\powk{k}{J}, 1\leq k\leq m$ is compatible with a uniform 
polymatroid $\Pol$, then the access structure determined by $\Delta$ and $\Pol$ is 
ideal. 
\end{thm}

Let us notice that Theorem \ref{thm:singl:delta} shows that the access structures 
presented in Theorem \ref{thm:singleton} are ideal. Now we turn to the objects 
considered in Theorems \ref{thm:last:column} and  \ref{thm:last-but-one:column}.

\begin{cor}
All access structures determined by any uniform polymatroid $\Pol=(J,h,\bfg)$ with the 
increment sequence $\bfg =(g_i)_{i\in I_m}$ such that $|J|\geq 3$ and $g_0\geq g_1 = g_{m-1}>0$ are ideal.
\end{cor}

\begin{proof}
We want to prove that for every increasing family $\Delta\subseteq \pow{J}\setminus \set{\emptyset}$ that is compatible with $\Pol$ the access structure determined by 
$\Pol$ and $\Gamma$ is ideal. The assumption $g_0\geq g_1 = g_{m-1}>0$ combined with 
Theorems \ref{thm:last:column} and  \ref{thm:last-but-one:column} imply that $\min \Delta= \powk{1}{J}$ or $|\min\Delta|=1$. In the former case the claim follows from 
Theorem \ref{thm:fpxy}. In the latter case applying Theorem \ref{thm:singl:delta} 
completes the proof. 
\end{proof}

%================================================================
\section{Conclusion}
\label{sec:concl}
%================================================================

This paper is intended to initiate research on the access structures obtained from 
polymatroids. This choice is motivated by the fact that access structures determined by 
polymatroids are matroid ports, i.e., they satisfy a necessary condition to be ideal. 
In this paper our investigation is limited to uniform polymatroids. We are particularly 
interested in the hierarchical order on the set of participants determined by  the 
access structures considered here. Most of the results in the literature that is devoted 
to discussing this subject consider access structures which are compartmented or totally 
hierarchical. We showed that all non compartmented access structure with at least three 
parties considered in this work are partially hierarchical. It is worth pointing out 
that some examples of partially hierarchical access structures are presented by  
Farr\`as et al. \cite{FPXY}, but they are not determined by uniform polymatroids.
There is good reason to deal with uniform polymatroids. In contrast to general 
polymatroids, every uniform polymatroid determines ideal access structures. It follows 
from the fact that every uniform polymatroid is representable. This allows building 
simple extensions of such polymatroids, which are also representable. Then according to 
Remark \ref{rem:represent} the suitable access structures obtained from those 
polymatroids are ideal. 

The conditions presented in Section \ref{sec:acc:str} are used to prove Theorems 
\ref{thm:extreme:columns} and \ref{thm:last:row} which show that most of access 
structures obtained from uniform polymatroids are compartmented (they are placed in the 
cells D2 and B3 - D3 of Table \ref{tab:general}). The exact hierarchy in access 
structures in the cells A2, B2, D1-F1, E2-3 and F2-3 is described in 
Theorems \ref{thm:g2:zero} - \ref{thm:singleton}. 

The most diverse collection of objects contains the cell C2 where both compartmented and 
hierarchical access structures can be found but further precise investigation of that 
area is necessary. In general, the results presented here do not exhaust the 
topic and leaves space for further research.

\begin{con}
Let $\Pi=(P_x)_{x\in J}$ be a partition of a set of participants $P$ and let 
$\Pol=(J,h,\bfg)$ be uniform polymatroid with $2\leq\eta(\bfg)<m$. Additionally, let 
$\Delta\seq \pow{J}\setminus\set{\emptyset}$ be a  monotone increasing family with 
$\mu(\Delta)=1$ that is compatible with $\Pol$. The hierarchical order in $\Pi$ induced 
by $\GPPD$ is of the type $\ord{Y}{X}$ for a certain disjoint subsets $X,\ Y$ of $J$.
\end{con}

This conjecture is partially confirmed by Theorem \ref{thm:height} that states that 
every chain in hierarchical access structure contains 1 or 2 elements. This fact applies 
only to access structures induced by uniform polymatroids. For other polymatroids one 
can construct hierarchical access structures with chains of arbitrary length. 

Some multipartite access structures determined by uniform polymatroids contain 
redundant blocks or different blocks that are equivalent. We treat such objects as 
improperly constructed. Fortunately, they appear only as extreme cases (cf. Corollary 
\ref{cor:equiv} and Theorem \ref{thm:connected}).

The results presented in Section \ref{sec:hier:pord} do not depend on the particular 
values of the rank function of $\Pol$ (or equivalently the values of $\bfg$). 
The only impact on the hierarchy of the described structures have the sequence 
of signatures of differences of consecutive entries of $\bfg$. This observation 
is additionally confirmed by computer calculations which suggest the following unproved 
conjecture. 

\begin{con}
\label{con:depend:on:sig}
Let $\bfg=(g_i)_{i\in I_m}$ and $\bfg'=(g'_i)_{i\in I_m}$ be the increment 
sequences of uniform polymatroids $\Pol$ and $\Pol'$ with the ground set $J$, 
respectively  such that $\sgn(g_{i-1}-g_i)=\sgn(g'_{i-1}-g'_i)$ for all 
$i=1,\ldots,m$. If a monotone increasing family $\Delta$ is compatible with $\Pol$ 
and $\Pol'$, then the hierarchical preorders on $\Pi$ determined by 
$\Gamma(\Pi,\Pol,\Delta)$ and $\Gamma(\Pi,\Pol',\Delta)$ are equal.
\end{con}

Investigating which of the structures considered in this article are ideal is another 
open issue. A sufficient condition can be obtained by proving that the simple extension 
of a given uniform polymatroid is representable (cf. \cite[Corollary 6.7]{FMFP}). This 
idea has been used to show that the access structure discussed in Theorems 
\ref{thm:g2:zero} and \ref{thm:last:column} - \ref{thm:singleton} are ideal. By 
analyzing the structure of the vector space representation of the polymatroid, one can 
also prove the ideality of many other access structures. However, we cannot rule out the 
existence of non-ideal access structures derived from uniform polymatroids. In this case 
we have the following question.   
Is it true that upper bound for the information ratio of access structures obtained form 
uniform polymatroids can be significantly less than the upper bound for the information 
ratio of arbitrary matroid ports?
Let us recall, the information ratio of a secret sharing scheme is the ratio between the 
maximum length of the shares and the length of the secret with a finite domain of 
shares. The \emph{information ratio of an access structure} $\Gamma$ is the infimum of 
all information ratios taken over all secret sharing schemes  with the access structure 
$\Gamma$.

%==================================================
\section*{Appendix}
%==================================================

Table 2 presents hierarchical (pre)orders of access structures determined by uniform 
polymatroids $\Pol=(J,h,\bfg)$ where $J=\set{1,2,3,4}$. It is worth pointing out that 
types of orders are invariant with respect to permutations of elements of $J$, so 
monotony increasing families appearing in the table are representatives of invariant 
classes of the permutation group $S_4$ acting on $J$. For example, the monotone 
increasing families $\Delta_1$ and $\Delta_2$ such that 
$\min\Delta_1=\set{\set{1},\set{2,3}}$ and $\min\Delta_2=\set{\set{2},\set{3,4}}$ 
represent the same invariant class. Assuming that Conjecture \ref{con:depend:on:sig} is 
true, the Table presents a complete overview of hierarchical orders of all access 
structures obtained from uniform polymatroids $(J,h,\bfg)$ with $|J|=4$. 
If the monotonic family appearing in the first column is not compatible with the 
polymatroid represented by the values of $\bfg$ in the top rows, then the suitable cell 
of the table contains $-$. Otherwise, the types of (pre)orders are denoted according to 
the following key. 

\vspace{8mm}

\begin{flushleft} % T FIGURE - threshold
\begin{tikzpicture}
[elem/.style={circle,fill=black!50,draw=black,thick,inner sep=2pt,minimum size=5},
hook/.style={circle,fill=black!0,thin,inner sep=0pt}]
%\begin{scope}[scale=1/2]
{

\node(set) at (0.0, 0.0) [ellipse,draw] {$P_1 P_2 P_3 P_4$};
%caption
\node[hook,label=below: {$T:= \pord{J_4}{\emptyset}  $}] (caption1) at (0.1, -0.7)  {};  
}
%end{scope}
\end{tikzpicture}
\end{flushleft}

\vspace{-27.5mm}

\begin{center} % C FIGURE - compartmented
\begin{tikzpicture}
[elem/.style={circle,fill=black!50,draw=black,thick,inner sep=2pt,minimum size=5},
hook/.style={circle,fill=black!0,thin,inner sep=0pt}]
%\begin{scope}[scale=1/2]
{
\node[elem,label=above: $P_{1}$] (B1)  at (-0.5, 0.0) {} ;
\node[elem,label=above: $P_{2}$] (A1)  at (0.5, 0.0) {};
\node[elem,label=above: $P_{3}$] (A1)  at (1.5, 0.0) {};
\node[elem,label=above: $P_{4}$] (A1)  at (2.5, 0.0) {};
%caption
\node[hook,label=below: {\hspace{25mm}$C:= \pord{\emptyset}{J_4}  = \ord{\emptyset}{J_4}  $}] (caption1) at (0.1, -0.7)  {};  
}
%end{scope}
\end{tikzpicture}
\end{center}

\vspace{2mm}

\begin{flushleft} % I FIGURE - Ord*(234,1)
\begin{tikzpicture}
[elem/.style={circle,fill=black!50,draw=black,thick,inner sep=2pt,minimum size=5},
hook/.style={circle,fill=black!0,thin,inner sep=0pt}]
%\begin{scope}[scale=1/2]
{
\node[elem,label=above: $P_{1}$] (B1)  at (0.0, 1.5) {} ;

\node(set) at (0.0, 0.0) [ellipse,draw] {$P_2 P_3 P_4$};

%edges
\draw[->] (node cs:name=set)-- (node cs:name=B1);
%caption
\node[hook,label=below: {$I:= \pord{\{2,3,4\}}{\{1\}}  $}] (caption1) at (0.1, -0.7)  {};  
}
%end{scope}
\end{tikzpicture}
\end{flushleft}

\vspace{-42.5mm}

\begin{center} % M FIGURE - Ord(2 3 4,1)
\begin{tikzpicture}
[elem/.style={circle,fill=black!50,draw=black,thick,inner sep=2pt,minimum size=5},
hook/.style={circle,fill=black!0,thin,inner sep=0pt}]
%\begin{scope}[scale=1/2]
{
\node[elem,label=above: $P_{1}$] (B1)  at (1.0, 1.5) {} ;
\node[elem,label=below: $P_{4}$] (A1)  at (2, 0.0) {};
\node[elem,label=below: $P_{3}$] (A2)  at (1, 0.0) {};
\node[elem,label=below: $P_{2}$] (A3)  at (0, 0.0) {};
%edges
\draw[->] (node cs:name=A1)-- (node cs:name=B1);
\draw[->] (node cs:name=A2)-- (node cs:name=B1);
\draw[->] (node cs:name=A3)-- (node cs:name=B1);
%caption
\node[hook,label=below: {\hspace{25mm}$M:= \ord{\{2,3,4\}}{\{1\}}  $}] (caption1) at (0.1, -0.7)  {};  
}
%end{scope}
\end{tikzpicture}
\end{center}

\vspace{2mm}

\begin{flushleft} % V FIGURE - Ord*(3 4,1 2)
\begin{tikzpicture}
[elem/.style={circle,fill=black!50,draw=black,thick,inner sep=2pt,minimum size=5},
hook/.style={circle,fill=black!0,thin,inner sep=0pt}]
%\begin{scope}[scale=1/2]
{
\node[elem,label=above: $P_{1}$] (B1)  at (-0.5, 1.5) {} ;
\node[elem,label=above: $P_{2}$] (B2)  at (0.7, 1.5) {} ;

\node(set) at (0.1, 0.0) [ellipse,draw] {$P_3 P_4$};

%edges
\draw[<-] (node cs:name=B1)-- (node cs:name=set);
\draw[<-] (node cs:name=B2)-- (node cs:name=set);

%caption
\node[hook,label=below: {$V:= \pord{\{3,4\}}{\{1,2\}}  $}] (caption1) at (0.1, -0.7)  {};  
}
%end{scope}
\end{tikzpicture}
\end{flushleft}

\vspace{-42.0mm}

\begin{center} % K FIGURE - Ord(3 4,1 2)
\begin{tikzpicture}
[elem/.style={circle,fill=black!50,draw=black,thick,inner sep=2pt,minimum size=5},
hook/.style={circle,fill=black!0,thin,inner sep=0pt}]
%\begin{scope}[scale=1/2]
{
\node[elem,label=above: $P_{1}$] (B1)  at (1.5, 1.5) {} ;
\node[elem,label=above: $P_{2}$] (B2)  at (0.5, 1.5) {} ;
\node[elem,label=below: $P_{3}$] (A1)  at (1.5, 0.0) {};
\node[elem,label=below: $P_{4}$] (A2)  at (0.5, 0.0) {};
%edges
\draw[->] (node cs:name=A1)-- (node cs:name=B1);
\draw[->] (node cs:name=A1)-- (node cs:name=B2);
\draw[->] (node cs:name=A2)-- (node cs:name=B1);
\draw[->] (node cs:name=A2)-- (node cs:name=B2);
%caption
\node[hook,label=below: {\hspace{25mm}$K:= \ord{\{3,4\}}{\{1,2\}}  $}] (caption1) at (0.1, -0.7)  {};  
}
%end{scope}
\end{tikzpicture}
\end{center}

\vspace{2mm}

\begin{flushleft} % E FIGURE - Ord(4,1)
\begin{tikzpicture}
[elem/.style={circle,fill=black!50,draw=black,thick,inner sep=2pt,minimum size=5},
hook/.style={circle,fill=black!0,thin,inner sep=0pt}]
%\begin{scope}[scale=1/2]
{
\node[elem,label=above: $P_{1}$] (B1)  at (0, 1.5) {} ;
\node[elem,label=below: $P_{4}$] (A1)  at (0.0, 0.0) {};
\node[elem,label=below: $P_{2}$] (C1)  at (1.0, 0.75) {};
\node[elem,label=below: $P_{3}$] (C2)  at (2.0, 0.75) {};
%edges
\draw[->] (node cs:name=A1)-- (node cs:name=B1);
%caption
\node[hook,label=below: {$E:= \ord{\{4\}}{\{1\}}  $}] (caption1) at (0.1, -0.7)  {};  
}
%end{scope}
\end{tikzpicture}
\end{flushleft}

\vspace{-42.0mm}

\begin{center} % W FIGURE - Ord(4,1 2 3)
\begin{tikzpicture}
[elem/.style={circle,fill=black!50,draw=black,thick,inner sep=2pt,minimum size=5},
hook/.style={circle,fill=black!0,thin,inner sep=0pt}]
%\begin{scope}[scale=1/2]
{
\node[elem,label=above: \hspace{-5mm} $P_{1}$] (B1)  at (0, 1.5) {} ;
\node[elem,label=above: \hspace{1mm} $P_{2}$] (B2)  at (1, 1.5) {} ;
\node[elem,label=above: \hspace{1mm} $P_{3}$] (B3)  at (2, 1.5) {} ;
\node[elem,label=below: \hspace{2mm} $P_{4}$] (A1)  at (1.0, 0.0) {};
%edges
\draw[->] (node cs:name=A1)-- (node cs:name=B1);
\draw[->] (node cs:name=A1)-- (node cs:name=B2);
\draw[->] (node cs:name=A1)-- (node cs:name=B3);
%caption
\node[hook,label=below: {\hspace{60.0mm}$W:= \pord{\{4\}}{\{1,2,3\}}  = \ord{\{4\}}{\{1,2,3\}}  $}] (caption1) at (0.1, -0.7)  {};  
}
%end{scope}
\end{tikzpicture}
\end{center}

\newpage

\begin{landscape}
\vfill

\begin{table}[ht]
\caption{Access structures in the case $m=4$.}
\vspace{1mm}
\label{tab:case4}
\begin{tabular}{|c|p{65mm} c|c|c|c|c|c|c|c|c|c|c|c|c|c|c|c|}
\hline
  & & & 1 & 2 & 3 & 4 & 5 & 6 & 7 & 8 & 9 & 10 & 11 & 12 & 13 & 14 & 15 \\
\hline 
    &\cellcolor{gray!30} &\cellcolor{gray!30} $g_0$ & \cellcolor{gray!30} 1 &\cellcolor{gray!30} 2 & \cellcolor{gray!30} 1 & \cellcolor{gray!30} 3\cellcolor{gray!30} &\cellcolor{gray!30} 2 &\cellcolor{gray!30} 2 &\cellcolor{gray!30} 1 & \cellcolor{gray!30} 3 &\cellcolor{gray!30} 2 & \cellcolor{gray!30} 4 & \cellcolor{gray!30} 3\cellcolor{gray!30} &\cellcolor{gray!30} 3 &\cellcolor{gray!30} 2 &\cellcolor{gray!30} 2 & \cellcolor{gray!30}1  \\
\hhline{|>{\arrayrulecolor{black}}->{\arrayrulecolor{gray!30}}->{\arrayrulecolor{gray!30}}->{\arrayrulecolor{black}}->{\arrayrulecolor{black}}->{\arrayrulecolor{black}}->{\arrayrulecolor{black}}->{\arrayrulecolor{black}}->{\arrayrulecolor{black}}->{\arrayrulecolor{black}}->{\arrayrulecolor{black}}->{\arrayrulecolor{black}}->{\arrayrulecolor{black}}->{\arrayrulecolor{black}}->{\arrayrulecolor{black}}->{\arrayrulecolor{black}}->{\arrayrulecolor{black}}->{\arrayrulecolor{black}}|-|}
    &\cellcolor{gray!30} &\cellcolor{gray!30} $g_1$ & \cellcolor{gray!30} 0 &\cellcolor{gray!30} 1 & \cellcolor{gray!30} 1 & \cellcolor{gray!30} 2\cellcolor{gray!30} &\cellcolor{gray!30} 2 &\cellcolor{gray!30} 1 &\cellcolor{gray!30} 1 & \cellcolor{gray!30} 2 &\cellcolor{gray!30} 2 & \cellcolor{gray!30} 3 & \cellcolor{gray!30} 3\cellcolor{gray!30} &\cellcolor{gray!30} 2 &\cellcolor{gray!30} 2 &\cellcolor{gray!30} 1 & \cellcolor{gray!30}1  \\
\hhline{|>{\arrayrulecolor{black}}->{\arrayrulecolor{gray!30}}->{\arrayrulecolor{gray!30}}->{\arrayrulecolor{black}}->{\arrayrulecolor{black}}->{\arrayrulecolor{black}}->{\arrayrulecolor{black}}->{\arrayrulecolor{black}}->{\arrayrulecolor{black}}->{\arrayrulecolor{black}}->{\arrayrulecolor{black}}->{\arrayrulecolor{black}}->{\arrayrulecolor{black}}->{\arrayrulecolor{black}}->{\arrayrulecolor{black}}->{\arrayrulecolor{black}}->{\arrayrulecolor{black}}->{\arrayrulecolor{black}}|-|}
    &\cellcolor{gray!30} &\cellcolor{gray!30} $g_2$ & \cellcolor{gray!30} 0 &\cellcolor{gray!30} 0 & \cellcolor{gray!30} 0 & \cellcolor{gray!30} 1\cellcolor{gray!30} &\cellcolor{gray!30} 1 &\cellcolor{gray!30} 1 &\cellcolor{gray!30} 1 & \cellcolor{gray!30} 1 &\cellcolor{gray!30} 1 & \cellcolor{gray!30} 2 & \cellcolor{gray!30} 2\cellcolor{gray!30} &\cellcolor{gray!30} 2 &\cellcolor{gray!30} 2 &\cellcolor{gray!30} 1 & \cellcolor{gray!30}1  \\
\hhline{|>{\arrayrulecolor{black}}->{\arrayrulecolor{gray!30}}->{\arrayrulecolor{gray!30}}->{\arrayrulecolor{black}}->{\arrayrulecolor{black}}->{\arrayrulecolor{black}}->{\arrayrulecolor{black}}->{\arrayrulecolor{black}}->{\arrayrulecolor{black}}->{\arrayrulecolor{black}}->{\arrayrulecolor{black}}->{\arrayrulecolor{black}}->{\arrayrulecolor{black}}->{\arrayrulecolor{black}}->{\arrayrulecolor{black}}->{\arrayrulecolor{black}}->{\arrayrulecolor{black}}->{\arrayrulecolor{black}}|-|} 
   &\cellcolor{gray!30} \diag{4}{65mm}{$\ \min\Delta$} &\cellcolor{gray!30} $g_3$ & \cellcolor{gray!30} 0 &\cellcolor{gray!30} 0 & \cellcolor{gray!30} 0 & \cellcolor{gray!30} 0\cellcolor{gray!30} &\cellcolor{gray!30} 0 &\cellcolor{gray!30} 0 &\cellcolor{gray!30} 0 & \cellcolor{gray!30} 1 &\cellcolor{gray!30} 1 & \cellcolor{gray!30} 1 & \cellcolor{gray!30} 1\cellcolor{gray!30} &\cellcolor{gray!30} 1 &\cellcolor{gray!30} 1 &\cellcolor{gray!30} 1 & \cellcolor{gray!30}1  \\
\hline 
\ 1 \rule[-2mm]{0mm}{.7cm}& \cellcolor{gray!30} {$\{\{1\}\}$} &\cellcolor{gray!30} & -- & --  & --  & -- & -- & -- & --  & $M$ & $M$ & $M$  & $M$ & $M$ & $M$ & $M$ & $I$\\
\hline
\ 2 \rule[-2mm]{0mm}{.7cm}& \cellcolor{gray!30} {$\{\{1\},\{2\}\}$} &\cellcolor{gray!30} & -- & --  & --  & $C$ & -- & -- & --  & -- & --  & $C$  & -- & -- & -- & -- & --\\
\hline
\ 3 \rule[-2mm]{0mm}{.7cm}& \cellcolor{gray!30} {$\{\{1\}, \{2\}, \{3\}\}$} & \cellcolor{gray!30} & -- & $W$  & --  & $C$ & -- & -- & --  & $C$ & --  & $C$  & -- & -- & -- & -- & --\\
\hline
\ 4 \rule[-2mm]{0mm}{.7cm}& \cellcolor{gray!30} {$\{\{1\}, \{2\}, \{3\}, \{4\}\}$} &\cellcolor{gray!30} & $T$ & $C$  & --  & $C$ & -- & $C$ & --  & $C$ & --  & $C$  & -- & $C$ & -- & $C$ & --\\
\hline
\ 5 \rule[-2mm]{0mm}{.7cm}& \cellcolor{gray!30} {$\{\{1\}, \{2\}, \{3,4\}\}$} &\cellcolor{gray!30} & -- & $K$  & --  & $C$ & -- & -- & --  & $C$ & --  & $C$  & -- & -- & -- & -- & --\\
\hline
\ 6 \rule[-2mm]{0mm}{.7cm}& \cellcolor{gray!30} {$\{\{1\}, \{2,3\}\}$} &\cellcolor{gray!30} & -- & --  & --  & $E$ & $E$ & -- & --  & -- & --  & $C$  & $C$ & -- & -- & -- & --\\
\hline
\ 7 \rule[-2mm]{0mm}{.7cm}& \cellcolor{gray!30} {$\{\{1\}, \{2,3\}, \{2,4\}\}$} &\cellcolor{gray!30} & -- & --  & --  & $C$ & $C$ & -- & --  & -- & --  & $C$  & $C$ & -- & -- & -- & --\\
\hline
\ 8 \rule[-2mm]{0mm}{.7cm}& \cellcolor{gray!30} {$\{\{1\}, \{2,3\}, \{2,4\}, \{3,4\}\}$} &\cellcolor{gray!30} & -- & $M$ & $M$  & $C$ & $C$ & -- & --  & $C$ & $C$  & $C$  & $C$ & -- & -- & -- & --\\
\hline
\ 9 \rule[-2mm]{0mm}{.7cm}& \cellcolor{gray!30} {$\{\{1\}, \{2,3,4\}\}$} &\cellcolor{gray!30} & -- & --  & --  & $M$ & $M$ & $M$ & $M$  & -- & --  & $C$  & $C$ & $C$ & $C$ & -- & --\\
\hline
\ 10 \rule[-2mm]{0mm}{.7cm}& \cellcolor{gray!30} {$\{\{1,2\}\}$} &\cellcolor{gray!30} & -- & --  & --  & -- & -- & -- & --  & $C$ & $C$  & $C$  & $C$ & $C$ & $C$ & $K$ & $V$\\
\hline
\ 11 \rule[-2mm]{0mm}{.7cm}& \cellcolor{gray!30} {$\{\{1,2\},\{1,3\}\}$} &\cellcolor{gray!30} & -- & --  & --  & -- & -- & -- & --  & -- & --  & $C$  & $C$ & -- & -- & -- & --\\
\hline
\ 12 \rule[-2mm]{0mm}{.7cm}& \cellcolor{gray!30} {$\{\{1,2\},\{3,4\}\}$} &\cellcolor{gray!30} & -- & --  & --  &  $C$  &  $C$  & $C$ & $C$  & -- & --  & $C$ & $C$ & $C$ & $C$ & -- & --\\
\hline
\ 13 \rule[-2mm]{0mm}{.7cm}& \cellcolor{gray!30} {$\{\{1,2\},\{1,3\},\{1,4\}\}$} &\cellcolor{gray!30} & -- & --  & --  & -- & -- & -- & --  & $C$ & $C$ & $C$ & $C$ & -- & -- & -- & --\\
\hline
\ 14 \rule[-2mm]{0mm}{.7cm}& \cellcolor{gray!30} {$\{\{1,2\},\{1,3\},\{2,3\}\}$} &\cellcolor{gray!30} & -- & --  & --  & $C$ & $C$ & -- & --  & -- & --  & $C$  & $C$ & -- & -- & -- & --\\
\hline
\end{tabular}
\end{table}

\end{landscape}

\newpage

\begin{landscape}

\begin{table}[ht]
\caption{Access structures in the case $m=4$.}
\vspace{1mm}

\begin{tabular}{|c|p{65mm} c|c|c|c|c|c|c|c|c|c|c|c|c|c|c|c|}
\hline
  & & & 1 & 2 & 3 & 4 & 5 & 6 & 7 & 8 & 9 & 10 & 11 & 12 & 13 & 14 & 15 \\
\hline
    &\cellcolor{gray!30} &\cellcolor{gray!30} $g_0$ & \cellcolor{gray!30} 1 &\cellcolor{gray!30} 2 & \cellcolor{gray!30} 1 & \cellcolor{gray!30} 3\cellcolor{gray!30} &\cellcolor{gray!30} 2 &\cellcolor{gray!30} 2 &\cellcolor{gray!30} 1 & \cellcolor{gray!30} 3 &\cellcolor{gray!30} 2 & \cellcolor{gray!30} 4 & \cellcolor{gray!30} 3\cellcolor{gray!30} &\cellcolor{gray!30} 3 &\cellcolor{gray!30} 2 &\cellcolor{gray!30} 2 & \cellcolor{gray!30}1  \\
\hhline{|>{\arrayrulecolor{black}}->{\arrayrulecolor{gray!30}}->{\arrayrulecolor{gray!30}}->{\arrayrulecolor{black}}->{\arrayrulecolor{black}}->{\arrayrulecolor{black}}->{\arrayrulecolor{black}}->{\arrayrulecolor{black}}->{\arrayrulecolor{black}}->{\arrayrulecolor{black}}->{\arrayrulecolor{black}}->{\arrayrulecolor{black}}->{\arrayrulecolor{black}}->{\arrayrulecolor{black}}->{\arrayrulecolor{black}}->{\arrayrulecolor{black}}->{\arrayrulecolor{black}}->{\arrayrulecolor{black}}|-|}
    &\cellcolor{gray!30} &\cellcolor{gray!30} $g_1$ & \cellcolor{gray!30} 0 &\cellcolor{gray!30} 1 & \cellcolor{gray!30} 1 & \cellcolor{gray!30} 2\cellcolor{gray!30} &\cellcolor{gray!30} 2 &\cellcolor{gray!30} 1 &\cellcolor{gray!30} 1 & \cellcolor{gray!30} 2 &\cellcolor{gray!30} 2 & \cellcolor{gray!30} 3 & \cellcolor{gray!30} 3\cellcolor{gray!30} &\cellcolor{gray!30} 2 &\cellcolor{gray!30} 2 &\cellcolor{gray!30} 1 & \cellcolor{gray!30}1  \\
\hhline{|>{\arrayrulecolor{black}}->{\arrayrulecolor{gray!30}}->{\arrayrulecolor{gray!30}}->{\arrayrulecolor{black}}->{\arrayrulecolor{black}}->{\arrayrulecolor{black}}->{\arrayrulecolor{black}}->{\arrayrulecolor{black}}->{\arrayrulecolor{black}}->{\arrayrulecolor{black}}->{\arrayrulecolor{black}}->{\arrayrulecolor{black}}->{\arrayrulecolor{black}}->{\arrayrulecolor{black}}->{\arrayrulecolor{black}}->{\arrayrulecolor{black}}->{\arrayrulecolor{black}}->{\arrayrulecolor{black}}|-|}
    &\cellcolor{gray!30} &\cellcolor{gray!30} $g_2$ & \cellcolor{gray!30} 0 &\cellcolor{gray!30} 0 & \cellcolor{gray!30} 0 & \cellcolor{gray!30} 1\cellcolor{gray!30} &\cellcolor{gray!30} 1 &\cellcolor{gray!30} 1 &\cellcolor{gray!30} 1 & \cellcolor{gray!30} 1 &\cellcolor{gray!30} 1 & \cellcolor{gray!30} 2 & \cellcolor{gray!30} 2\cellcolor{gray!30} &\cellcolor{gray!30} 2 &\cellcolor{gray!30} 2 &\cellcolor{gray!30} 1 & \cellcolor{gray!30}1  \\
\hhline{|>{\arrayrulecolor{black}}->{\arrayrulecolor{gray!30}}->{\arrayrulecolor{gray!30}}->{\arrayrulecolor{black}}->{\arrayrulecolor{black}}->{\arrayrulecolor{black}}->{\arrayrulecolor{black}}->{\arrayrulecolor{black}}->{\arrayrulecolor{black}}->{\arrayrulecolor{black}}->{\arrayrulecolor{black}}->{\arrayrulecolor{black}}->{\arrayrulecolor{black}}->{\arrayrulecolor{black}}->{\arrayrulecolor{black}}->{\arrayrulecolor{black}}->{\arrayrulecolor{black}}->{\arrayrulecolor{black}}|-|} 
    &\cellcolor{gray!30} \diag{4}{65mm}{$\ \min\Delta$} &\cellcolor{gray!30} $g_3$ & \cellcolor{gray!30} 0 &\cellcolor{gray!30} 0 & \cellcolor{gray!30} 0 & \cellcolor{gray!30} 0\cellcolor{gray!30} &\cellcolor{gray!30} 0 &\cellcolor{gray!30} 0 &\cellcolor{gray!30} 0 & \cellcolor{gray!30} 1 &\cellcolor{gray!30} 1 & \cellcolor{gray!30} 1 & \cellcolor{gray!30} 1\cellcolor{gray!30} &\cellcolor{gray!30} 1 &\cellcolor{gray!30} 1 &\cellcolor{gray!30} 1 & \cellcolor{gray!30}1  \\
\hline 
\ 15 \rule[-2mm]{0mm}{.7cm}& \cellcolor{gray!30} {$\{\{1,2\},\{2,3\},\{1,4\}\}$} &\cellcolor{gray!30} & -- & --  & --  & $C$ & $C$ & -- & --  & -- & --  & $C$  & $C$ & -- & -- & -- & --\\
\hline
\ 16 \rule[-2mm]{0mm}{.7cm}& \cellcolor{gray!30} {$\{\{1,3\},\{2,3\},\{1,4\},\{2,4\}\}$}&\cellcolor{gray!30}  & -- & --  & --  & $C$ & $C$ & -- & --  & -- & --  & $C$  & $C$ & -- & -- & -- & --\\
\hline
\ 17 \rule[-2mm]{0mm}{.7cm}& \cellcolor{gray!30} {$\{\{1,2\},\{1,3\},\{2,3\} \{1,4\}\}$}&\cellcolor{gray!30}  & -- & --  & --  & $C$ & $C$ & -- & --  & -- & --  & $C$  & $C$ & -- & -- & -- & --\\
\hline
\ 18 \rule[-2mm]{0mm}{.7cm}& \cellcolor{gray!30} {$\{\{1,2\},\{1,3\},\{2,3\}, \{1,4\}, \{2,4\}\}$}&\cellcolor{gray!30} & -- & --  & --  & $C$ & $C$ & -- & --  & -- & --  & $C$  & $C$ & -- & -- & -- & --\\
\hline
\ 19 \rule[-2mm]{0mm}{.7cm}& \cellcolor{gray!30} {$\{\{1,2\},\{1,3\},\{2,3\}, \{1,4\}, \{2,4\}, \{3,4\}\}$} &\cellcolor{gray!30}& -- & $C$  & $C$  & $C$ & $C$ & -- & --  & $C$ & $C$  & $C$  & $C$ & -- & -- & -- & --\\ 
\hline
\ 20 \rule[-2mm]{0mm}{.7cm}& \cellcolor{gray!30} {$\{\{1,2\}, \{1,3,4\}\}$} &\cellcolor{gray!30} & -- & --  & --  & -- & -- & -- & --  & -- & --  & $C$  & $C$ & $C$ & $C$ & -- & --\\
\hline
\ 21 \rule[-2mm]{0mm}{.7cm}& \cellcolor{gray!30} {$\{\{1,2\}, \{1,3\}, \{2,3,4\}\}$}&\cellcolor{gray!30}  & -- & --  & --  & $C$ & $C$ & -- & --  & -- & --  & $C$ & $C$ & -- & -- & -- & --\\
\hline
\ 22 \rule[-2mm]{0mm}{.7cm}& \cellcolor{gray!30} {$\{\{1,2\}, \{1,3\}, \{1,4\}, \{2,3,4\}\}$}&\cellcolor{gray!30}  & -- & --  & --  & $C$ & $C$ & -- & --  & -- & --  & $C$ & $C$ & -- & -- & -- & --\\
\hline
\ 23 \rule[-2mm]{0mm}{.7cm}& \cellcolor{gray!30} {$\{\{1,2\}, \{1,3,4\},\{2,3,4\}\}$}&\cellcolor{gray!30}  & -- & --  & --  & $C$ & $C$ & $C$ & $C$  & -- & --  & $C$ & $C$ & $C$ & $C$ & -- & --\\
\hline
\ 24 \rule[-2mm]{0mm}{.7cm}& \cellcolor{gray!30} {$\{\{1,2,3\}\}$} &\cellcolor{gray!30} & -- & --  & --  & -- & -- & -- & --  & $C$ & $C$  & $C$  & $C$ & $C$ & $C$ & $W$ & $W$\\
\hline
\ 25 \rule[-2mm]{0mm}{.7cm}& \cellcolor{gray!30} {$\{\{1,2,3\},\{1,2,4\} \}$}&\cellcolor{gray!30}  & -- & --  & --  & -- & -- & -- & --  & -- & --  & $C$  & $C$ & $C$ & $C$ & -- & --\\
\hline
\ 26 \rule[-2mm]{0mm}{.7cm}& \cellcolor{gray!30} {$\{\{1,2,3\},\{1,2,4\},\{1,3,4\}\}$}&\cellcolor{gray!30}  & -- & --  & --  & -- & -- & -- & --  & -- & --  & $C$  & $C$ & $C$ & $C$ & -- & --\\
\hline
\ 27 \rule[-2mm]{0mm}{.7cm}& \cellcolor{gray!30} {$\{\{1,2,3\},\{1,2,4\},\{1,3,4\},\{2,3,4\}\}$}&\cellcolor{gray!30} & -- & --  & --  & $C$ & $C$ & $C$ & $C$  & -- & --  & $C$  & $C$ & $C$ & $C$ & -- & --\\
\hline
\ 28 \rule[-2mm]{0mm}{.7cm}& \cellcolor{gray!30} {$\{\{1,2,3,4\}\}$} &\cellcolor{gray!30} & -- & --  & --  & -- & -- & -- & --  & $C$ & $C$  & $C$  & $C$ & $C$ & $C$ & $C$ & $C$\\
\hline
\end{tabular}
\end{table}

\end{landscape}

\end{document}